\documentclass[12pt,letterpaper]{article}

\usepackage[ left=1in, top=1in, right=1in, bottom=1in]{geometry}
\usepackage{graphicx,bm,colonequals,amsmath,amssymb,url,xcolor,bbm}
\usepackage{array,tabularx,multirow}
\usepackage{enumitem}
\usepackage[font={footnotesize}]{caption,subcaption}
\usepackage[utf8]{inputenc}
\usepackage{enumitem}
\usepackage{soul} % for strikethrough text
\usepackage{placeins} % for FloatBarrier
\usepackage[normalem]{ulem}
\usepackage{natbib,hyperref}

\graphicspath{{plots/}}

\usepackage{mathtools}
\mathtoolsset{showonlyrefs} 

\bibpunct[, ]{(}{)}{;}{a}{,}{,}
   
\usepackage{amsthm}
\newtheoremstyle{propstyle} % name
    {2mm}                    % Space above
    {1mm}                    % Space below
    {\itshape}                   % Body font
    {}                           % Indent amount
    {\scshape}                   % Theorem head font
    {.}                          % Punctuation after theorem head
    {.5em}                       % Space after theorem head
    {}  % Theorem head spec (can be left empty, meaning ‘normal’)
\theoremstyle{propstyle}
\newtheorem{proposition}{Proposition}
\theoremstyle{propstyle}

\theoremstyle{propstyle}

\theoremstyle{propstyle}

\theoremstyle{propstyle}

\usepackage[ruled,vlined]{algorithm2e}
\SetKwInput{KwInput}{Input}
\usepackage{algorithmic}
\usepackage{array,framed}
\usepackage{float}

\makeatletter
\renewcommand{\paragraph}{%
  \@startsection{paragraph}{4}%
  {\z@}{2ex \@plus 1ex \@minus .2ex}{-1em}%
  {\normalfont\normalsize\bfseries}%
}
\makeatother

\DeclareMathAlphabet\mathbfcal{OMS}{cmsy}{b}{n}

\newcommand{\bd}{\mathbf{d}}
\renewcommand{\bf}{\mathbf{f}}
\newcommand{\bj}{\mathbf{j}}

\newcommand{\bs}{\mathbf{s}}

\newcommand{\by}{\mathbf{y}}

\newcommand{\bk}{\mathbf{k}}

\newcommand{\bw}{\mathbf{w}}

\newcommand{\bz}{\mathbf{z}}

\newcommand{\bY}{\mathbf{Y}}

\newcommand{\bF}{\mathbf{F}}
\newcommand{\bG}{\mathbf{G}}
\newcommand{\bE}{\mathbf{E}}

\newcommand{\bI}{\mathbf{I}}

\newcommand{\bV}{\mathbf{V}}
\newcommand{\bK}{\mathbf{K}}

\newcommand{\bQ}{\mathbf{Q}}

\newcommand{\bfzero}{\mathbf{0}}

\newcommand{\bfmu}{\bm{\mu}}
\newcommand{\bfxi}{\bm{\xi}}
\newcommand{\bftheta}{\bm{\theta}}

\newcommand{\bfkappa}{\bm{\kappa}}

\newcommand{\bfpsi}{\bm{\psi}}

\newcommand{\bfSigma}{\bm{\Sigma}}

\newcommand{\diag}{diag}

\newcommand{\GP}{\mathcal{GP}}

\newcommand{\order}{\mathcal{O}}
\newcommand{\normal}{\mathcal{N}}

\newcommand{\map}{\mathcal{T}}
\newcommand{\pmap}{\widetilde{\mathcal{T}}}

\title{Learning non-Gaussian spatial distributions via Bayesian transport maps with parametric shrinkage}

\author{Anirban Chakraborty\thanks{Department of Bioinformatics and Computational Biology, University of Texas MD Anderson Cancer Center} \and Matthias Katzfuss\thanks{Department of Statistics, University of Wisconsin--Madison. Corresponding author: \texttt{katzfuss@gmail.com}.}}

\date{}

\begin{document}
\maketitle
\begin{abstract}
Many applications, including climate-model analysis and stochastic weather generators, require learning or emulating the distribution of a high-dimensional and non-Gaussian spatial field based on relatively few training samples. To address this challenge, a recently proposed Bayesian transport map (BTM) approach consists of a triangular transport map with nonparametric Gaussian-process (GP) components, which is trained to transform the distribution of interest to a Gaussian reference distribution. To improve the performance of this existing BTM, we propose to shrink the map components toward a ``base'' parametric Gaussian family combined with a Vecchia approximation for scalability. The resulting ShrinkTM approach is more accurate than the existing BTM, especially for small numbers of training samples. It can even outperform the ``base'' family when trained on a single sample of the spatial field. We demonstrate the advantage of ShrinkTM through numerical experiments on simulated data and on climate-model output. 
\end{abstract}


%%%%%%%%%

%%%%%%%%%%%%%%%%%%%%%%%%%%%%%%%%%%%%%%%%%%%%%%%%%%%%%%%

\section{Introduction \label{sec:intro}}

% motivation
Statistical inference in several practical applications, including climate-model emulation \citep[e.g.,][]{Castruccio2014,Nychka2018,Haugen2019} and ensemble-based data assimilation \citep[e.g.,][]{Houtekamer2016,Katzfuss2015b}, requires generative modeling of high-dimensional spatial distributions based on small number of replicates that are expensive to produce. Several methods are available for spatial inference with a small number of training replicates, including Gaussian processes (GP) with simple parametric covariance functions \citep[e.g.,][]{Cressie1993,Banerjee2004}, locally anisotropic Mat\'ern covariances \citep[e.g.,][]{Nychka2018,Wiens2020}, and copulas \citep[e.g.,][]{Graler2014}. However, many of these methods assume Gaussianity implicitly or explicitly in the model and do not scale to large datasets \citep[see][for a review]{Katzfuss2021}.

Generative machine-learning based approaches such as variational auto encoders, generative adversarial networks \citep[e.g.,][]{Goodfellow2016}, and normalizing flows \citep[e.g.,][]{Kobyzev2020normalizing} on the other hand %have been used for climate-model emulation \citep[e.g.,][]{Ayala2021,Besombes2021} due to their highly expressive nature. However, these models 
typically require massive training data, and are often highly sensitive to tuning-parameter and network-architecture choices \citep[e.g.,][]{Arjovsky2017,Hestness2017,Mescheder2018}, and hence not directly applicable to low-data applications without additional application-specific techniques \citep[e.g.,][]{Kashinath2021}.

Transport maps constitute another approach for modeling continuous multivariate non-Gaussian distributions. Specifically, a triangular transport map can characterize the dependence structure of any continuous multivariate distribution by transforming the target distribution to a reference distribution (e.g., standard Gaussian) \citep[e.g.,][]{Marzouk2016, Katzfuss2021}. Recently, \citet{Katzfuss2021} introduced a Bayesian transport map (BTM), in which the components of the transport map are nonparametrically modeled as GPs. \citet{Wiemann2023BayesianFields} extended this approach to model multivariate non-Gaussian spatial fields. The BTM approach results in closed-form inference that quantifies uncertainty and avoids under- and over-fitting even when the number of training samples is small. For target distributions corresponding to spatial fields, \citet{Katzfuss2021} proposed priors that exploit the screening effect via suitable conditional-independence assumptions that guarantee computational scalability for very large datasets. The resulting sparse non-linear transport maps can be seen as a non-parametric and non-Gaussian generalization of Vecchia approximations \citep[e.g.,][]{Vecchia1988, Stein2004, Datta2016, Katzfuss2017a, Schafer2020}, which implicitly utilize a linear transport map given by a sparse inverse Cholesky factor. The BTM approach, however, still requires $\order(10)$ training samples to accurately estimate meaningful dependencies, and hence can be prohibitive in applications where obtaining even a moderate number of training ensembles can become quite expensive.

We propose a novel extension of the BTM approach of \citet{Katzfuss2021} that is specifically tailored for learning the distributions of large spatial fields when very limited training samples (often, only one) are available. This extension builds on \citet{Kidd2020}, who proposed Bayesian nonparametric inference on the inverse Cholesky factors of the GP covariance matrices and empirically studied their parametric regularization. The major contribution lies in defining a new set of prior distributions for the map components that first centers the mean and variance of each of the components toward the conditional means and variances of a parametric GP (e.g., with a Mat\'ern covariance). Next, it introduces some regularization factors in the prior distributions, which can be optimized to actively learn the amount of shrinkage towards the GP family. We approximate the conditional means and variances by leveraging Vecchia approximation \citep{Vecchia1988}. The resulting regularized BTM approach, which we call ShrinkTM, is fast and produces accurate inference for large spatial non-Gaussian fields even when trained with only $n=1$ sample. Consequently, it expands the application of the BTM approach introduced by \citet{Katzfuss2021} in many realistic applications.

The remainder of the paper is organized as follows. In Section \ref{sec:btmreview}, we review the Bayesian spatial transport map of \citet{Katzfuss2021}. In Section \ref{sec:methodology}, we introduce our regularization for the BTM approach, including the use of the Vecchia approximation for computing conditional means and variances of the base GP family. In Section \ref{sec:numerical}, we provide numerical comparisons of our methods to the existing BTM approach using simulated data and climate-model output. We conclude in Section \ref{sec:conclusion}. Proofs are provided in Appendix~\ref{app:proofs}.

%%%%%%%%%%%%%%%%%%%%%%%%%%%%%%%%%%%%%%%%%%%%%%%%%%%%%%%

%\section{Methodology \label{sec:methodology}}
\section{Review of Bayesian spatial transport maps \label{sec:btmreview}}

\subsection{Transport maps and regression\label{sec:tmreg}}

Our goal is to learn the joint probability distribution $p(\by)$ of a centered spatial field $\by = (y_1,\ldots,y_N)^\top$, where $y_i = y(\bs_i)$ is observed at location $\bs_i$. A transport map $\map: \mathbb{R}^N \to \mathbb{R}^N$ characterizes $p(\by)$ through a nonlinear transformation $\map$ of $\by$, such that $\map(\by)$ follows a simple reference distribution: $\map(\by) \sim \normal_N(\bm{0}, \bm{I}_N)$, where $\normal_N$ denotes an $N$-variate Gaussian distribution. Without loss of generality, $\map$ can have a lower triangular form \citep{Rosenblatt1952, Carlier2009} such that
\begin{equation}
\map(\by) = \begin{bmatrix*}[l] \map_1(y_1) \\ \map_2(y_1,y_2) \\ ~\,\vdots \\ \map_N(y_1,y_2,\ldots,y_N) \end{bmatrix*}, \label{eq:origtmap}
\end{equation}
with $\mathcal{T}_i$ being strictly monotone in the $i$-th argument. \citet{Katzfuss2021} have used this transport-map idea for generative modeling of high-dimensional spatial fields. Specifically, they assume map components $\map_i$ of the form,
\begin{equation}
    \map_i(\by_{1:i}) = (y_i - f_i(\by_{1:i-1}))/d_i, \qquad i=1,\ldots,N,
\end{equation}
for some $d_i \in \mathbb{R}^+$, $f_i: \mathbb{R}^{i-1} \rightarrow \mathbb{R}$ for $i=2,\ldots,N$, and $f_i(\by_{1:i-1}) \equiv 0$ for $i=1$. 
Consequently, the implied joint density $p(\by)$ can be factorized as
$
p(\by) = \prod_{i=1}^N \normal\big(y_i|f_i(\by_{1:i-1}),d_i\big),
$
and so the transport-map approach converts the problem of inferring the $N$-variate distribution of $\by$ into $N$ independent regressions of $y_i$ on $\by_{1:i-1}$ of the form
\begin{equation}
y_i = f_i(\by_{1:i-1}) + \epsilon_i, \quad \epsilon_i \sim \normal(0,d_i^2), \qquad i=1,\ldots,N.  \label{eq:origtmapreform}
\end{equation}

\subsection{Prior distributions}
In contrast to previous transport map literature \citep[e.g.,][]{Marzouk2016}, which assumes known parametric forms for $\map$, \citet{Katzfuss2021} assume a flexible, nonparametric prior on the map $\map$ by specifying independent conjugate Gaussian-process-inverse-Gamma priors for the $f_i$ and $d_i^2$. Specifically, for the conditional variances $d_i^2$, they assume inverse-Gamma distributions, 
\begin{equation}
\label{eq:dprior}
    d_i^2 \stackrel{ind.}{\sim} \mathcal{IG}(\alpha_i,\beta_i), \qquad \text{with } \alpha_i>1, \; \beta_i > 0, \qquad i=1,\ldots,N.
\end{equation}    
Conditional on $d_i^2$, each function $f_i$ is modeled as a Gaussian process (GP) with inputs $\by_{1:i-1}$,
\begin{equation}
    \label{eq:gp}
f_i | d_i \stackrel{ind.}{\sim} \GP(0,d_i^2 K_i), \qquad i=1,\ldots,N,
\end{equation}
with $K_i(\cdot,\cdot) = C_i(\cdot,\cdot)/E(d_i^2)$, $E(d_i^2) = \beta_i/(\alpha_i -1)$,
\begin{equation}
    \label{eq:kernel}
C_i(\by_{1:i-1},\by_{1:i-1}') = \by_{1:i-1}^\top\bQ_i\by_{1:i-1}^\prime + \sigma^2_i \, \rho(h_i(\by_{1:i-1}^\top, \by^\prime_{1:i-1})/\gamma), \qquad i=2,\ldots,N,
\end{equation}
where $\sigma_i \in \mathbb{R}^+_0$, $\gamma = \exp(\theta_\gamma)$ is a range parameter, $h_i^2(\by_{1:i-1},\by_{1:i-1}')$ = $(\by_{1:i-1}-\by_{1:i-1}')^\top\bQ_i(\by_{1:i-1}-\by_{1:i-1}')$ and $\rho$ is an isotropic correlation function such that $\rho(h_i(\by_{1:i-1},\by_{1:i-1})) =1$. The expression of prior covariance function $C_i(\cdot, \cdot)$ in \eqref{eq:kernel} takes into account variance from both linear and nonlinear components from the regression function $f_i$.
The degree of nonlinearity of $f_i$ is determined by $\sigma_i^2$, and indicates linear dependency of $f_i$ on $\by_{1:i-1}$ when $\sigma^2_i = 0$.

\subsection{Maximin ordering for large spatial fields \label{sec:maximin}}

\citet{Katzfuss2021} assume a maximum-minimum-distance (maximin) ordering of the locations $\bs_1,\ldots,\bs_N$, which is obtained by sequentially choosing each location to maximize the minimum distance to all previously ordered locations.
Specifically, the first index $i_1$ is chosen arbitrarily (e.g., $i_1=1$), and then the subsequent indices are selected as
$
i_j = \arg\max_{i \, \notin \, \mathcal{I}_{j}} \,\, \min_{j \, \in \, \mathcal{I}_{j}} \|\bs_i - \bs_j\|
$
for $j = 2, \ldots , N$, where $\mathcal{I}_{j} = \{i_1 , \ldots , i_{j-1}\}$. 
Define $c_i(k)$ as the index of the $k$th nearest (previously ordered) neighbor of the $i$th location (and so $\bs_{c_i(1)},\ldots,\bs_{c_i(4)}$ are indicated by ${\color{red}\mathbf{*}}$ in Figure \ref{fig:maxmin}).
\begin{figure}
\centering
	\begin{subfigure}{.24\textwidth}
	\centering
 	\includegraphics[width =.99\linewidth]{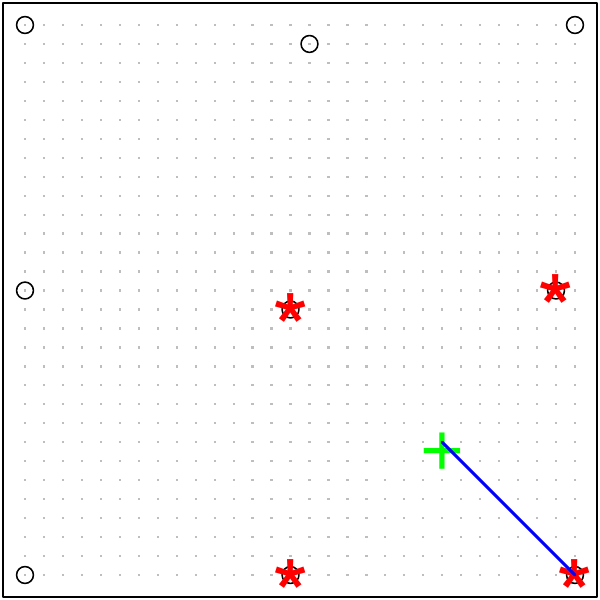}
	\caption{$i=10$}
	\label{fig:mm1}
	\end{subfigure}%
\hfill
\centering
	\begin{subfigure}{.24\textwidth}
	\centering
 	\includegraphics[width =.99\linewidth]{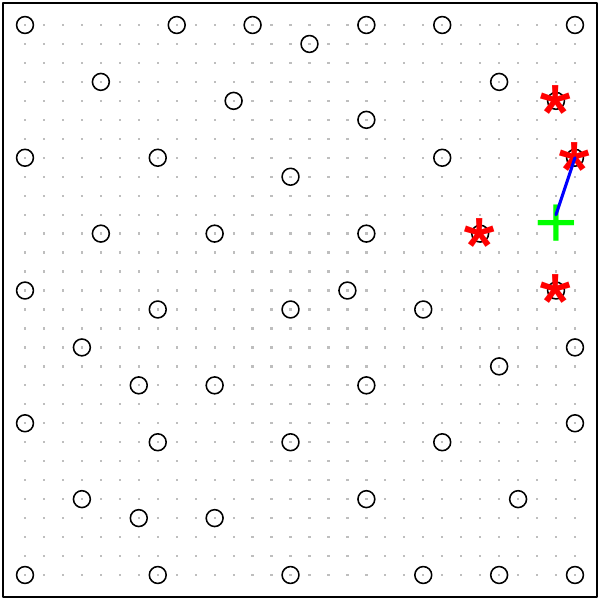}
	\caption{$i=50$}
	\label{fig:mm2}
	\end{subfigure}%
\hfill
\centering
	\begin{subfigure}{.24\textwidth}
	\centering
 	\includegraphics[width =.99\linewidth]{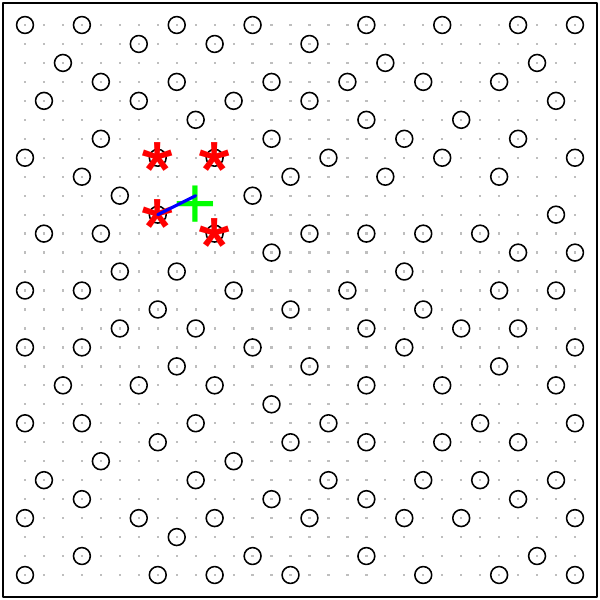}
	\caption{$i=131$}
	\label{fig:mm3}
	\end{subfigure}%
\hfill
	\begin{subfigure}{.24\textwidth}
	\centering
	\includegraphics[width =.99\linewidth]{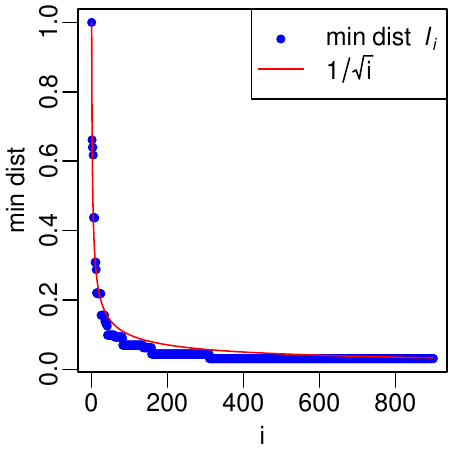}
	\caption{scale decay}
	\label{fig:scale}
	\end{subfigure}%
  \caption{Maximin ordering of locations on a grid (small gray points) of size $N=30 \times 30 = 900$ on $[0,1]^2$. (a)--(c): The $i$th ordered location (${\color{green}+}$), the previous $i-1$ locations (${\color{black}\circ}$), including the nearest $m=4$ neighbors (${\color{red}*}$ and the distance $\ell_i$ to the nearest neighbor ({\color{red}---}). (d): For $i=1,\ldots,N$, the length scales (i.e., minimum distances) decay as $\ell_i = i^{-1/\text{dim}}$, where dim denotes the dimension of the spatial domain. This figure closely follows Figure 2 in \citet{Katzfuss2021}.}
\label{fig:maxmin}
\end{figure}
The maximin ordering can be interpreted as a multiresolution decomposition into coarse scales early in the ordering and fine scales later in the ordering. In particular, the minimal pairwise distance $\ell_i = \| \bs_i - \bs_{c_i(1)} \|$ among the first $i$ locations of the ordering decays roughly as $\ell_i \propto i^{-1/\text{dim}}$, where $\text{dim}$ here is the dimension of the spatial domain (see Figure \ref{fig:scale}).

Assuming that the entries of $\by$ ordered according to maximin ordering, the $i$th regression in \eqref{eq:origtmapreform} can be viewed as a spatial prediction at location $\bs_i$ based on data at locations $\bs_1,\ldots,\bs_{i-1}$ that lie roughly on a regular grid with distance (i.e., scale) $\ell_i$. 

When the variables $y_1,\ldots,y_N$ are not associated with spatial locations or when Euclidean distance between the locations is not meaningful (e.g., nonstationary, multivariate, spatio-temporal, or functional data), the maximin and neighbor ordering can be carried out based on other distance metrics, such as $(1 - |\text{correlation}|)^{1/2}$ based on some guess or estimate of the correlation between variables \citep{Kidd2020, Kang2021}.

\subsection{Computation and inference for large spatial fields}
Inference based on \eqref{eq:origtmapreform} -- \eqref{eq:kernel} is computationally and statistically challenging for large $N$, as it requires inferring nonlinear functions $f_i$ in $\order(N)$ dimensions. Hence, \citet{Katzfuss2021} first assumed a maximin ordering of the spatial locations $s_1$, \ldots, $s_N$, and parameterized the relevance matrix $\bQ_i = \diag(q_{i,1}^2,\ldots,q_{i,i-1}^2)$ to decay exponentially with neighbor number $k$, as
\begin{equation}
q_{i,c_i(k)} = \begin{cases} \exp(-e^{\theta_q} k), & k \leq m^\prime,\\ 0, & k>m^\prime. \end{cases}
\label{eq:sparsity}    
\end{equation}
The sparsity parameter $m^\prime=\max \{k: \exp(-e^{\theta_q} k) \geq \varepsilon\}$ is determined by the data through the hyperparameter $\theta_q$, where $\varepsilon$ is a small fixed threshold (e.g., $\varepsilon=0.01$). The parametrization of $q_{i, c_i(k)}$ has been primarily considered by \citet{Katzfuss2021} and is motivated by exponential rates of screening for Gaussian processes derived from elliptic boundary-value problems \citep{Schafer2017,Schafer2020,Kang2024}. Consequently, $f_i(\by_{1:i-1})=f_i(\by_{g_{m^\prime}(i)})$ and its covariance kernel can be expressed as
\begin{equation}
    \label{eq:kernelreconstruct}
C_i(\by_{g_{m^\prime}(i)},\by_{g_{m^\prime}(i)}') = \by_{g_{m^\prime}(i)}^\top\bQ_i\by_{g_{m^\prime}(i)}^\prime + \sigma^2_i \, \rho(h_i(\by_{g_{m^\prime}(i)}^\top, \by^\prime_{g_{m^\prime}(i)})/\gamma), \qquad i=2,\ldots,N,
\end{equation}
where  $g_{m^\prime}(i) = \{ c_i(1),\ldots,c_i(m') \} \subset \{1, \ldots, i-1\}$ denotes the  $m^\prime$ nearest previously ordered neighbors. This results in a sparse $\GP$ regression of $f_i$ in \eqref{eq:origtmapreform} on $\by_{g_{m^{\prime}}(i)}$, which is easier to compute.

The prior distributions in equations \eqref{eq:dprior}-\eqref{eq:gp} and \eqref{eq:kernelreconstruct} depend on the hyperparameter vector $\bfpsi = (\alpha_1, \ldots, \alpha_N, \beta_1, \ldots, \beta_N, \bfkappa_1, \ldots, \bfkappa_N)$, where $\bfkappa_i$ are hyperparameters  of the covariance matrices $\bK$'s for the $\GP$'s on $f_i$'s. \citet{Katzfuss2021} have employed an empirical Bayes (EB) technique to estimate these hyperparameters. Specifically, they have shown that the \textit{screening effect} \citep[e.g.,][]{Stein2011} and \textit{near-Gaussianity} \citep[e.g.,][]{Schafer2021b, Katzfuss2021} along with the maximin ordering described in Section~\ref{sec:maximin} can be judiciously employed in the estimation process to simultaneously reduce the number of hyperparameters \citep[see][Sect. 2.1.1, for a review]{Wiemann2023BayesianFields} and decrease the computation time for large spatial fields. We will revisit relevant details discussing our methodology in Section \ref{sec:methodology}.

%%%%%%%%%%%%%
\section{Methodology: Shrinkage toward parametric covariances\label{sec:methodology}}

\subsection{Vecchia approximation of Gaussian distributions\label{sec:vecchia}}

Assume temporarily that $\by \sim \normal_N(\bfzero,\bfSigma)$ with known (spatial) covariance $\bfSigma$. In this setting, \citet{Vecchia1988} proposed to approximate the joint distribution $p(\by) = \prod_{i=1}^N p(y_i | \by_{1:i-1})$ by 
\begin{equation}
    \hat{p}(\by) = \prod_{i=1}^N p(y_i | \by_{g_m(i)}),
    \label{eq:vecchia}
\end{equation}
where $g_m(i) \subset \{1, \ldots, i-1\}$ is a subset of $m$ nearest neighbors ordered prior to $i$. The Vecchia approximation under maximin ordering has been shown to be highly accurate even for large $N$ (e.g., $N=10^5$) and relatively small $m$ (e.g., $m=30$) both numerically and theoretically in a variety of settings \citep[e.g.,][]{Datta2016,Katzfuss2017a, Katzfuss2018, Zilber2019, Jurek2020, Zhang2022,Kang2024}, due to the screening effect \citep[e.g.,][]{Stein2011}.

Due to the Gaussian assumption, the conditional distributions in \eqref{eq:vecchia} can be analytically expressed as:
\begin{equation}
    p(y_i | \by_{g_m(i)}) = \normal(y_i |\bfxi_i^\top\by_{g_m(i)},\tau^2_{i}),
    \label{eq:vecchiai}
\end{equation}
where $\bfxi_i^\top = \bfSigma_{i,g_m(i)} \bfSigma^{-1}_{g_m(i), g_m(i)}$, $\tau^2_{i}=\bfSigma_{i,i} - \bfSigma_{i,g_m(i)} \bfSigma^{-1}_{g_m(i), g_m(i)}\bfSigma_{g_m(i),i}$, and $\bfSigma_{\bj, \bk}$ denotes the submatrix of $\bfSigma$ with entries from row indices $\bj$ and column indices $\bk$ (with $\bj, \bk\subset \{1, 2, \ldots, N\}$).

Comparing \eqref{eq:vecchiai} to the expressions in Section~\ref{sec:tmreg}, we can see that if $\by\sim \normal(\bfzero, \bfSigma)$, both $f_i$ and $d_i^2$ can be expressed in closed form as $f_i(\by_{1:i-1}) = \bfxi_i^\top\by_{g_m(i)}$ and $d_i^2 = \tau^2_{i}$, after applying a Vecchia approximation to $\normal_N(\bfzero,\bfSigma)$. %This motivates us to design a prior that shrinks the posterior distribution towards a multivariate Gaussian distribution if the true data is generated from $\normal_N(\bfzero, \bfSigma)$.

%%%%
\subsection{Nonlinear prior distributions \label{sec:fdpriornew}}

Given the motivation in Section~\ref{sec:vecchia}, we now propose a methodology for learning the joint distribution $p(\by)$ with nonlinear and nonparametric dependence based on a small number of training samples by shrinking $p(\by)$ toward $\normal_N(\bfzero,\bfSigma)$ by shrinking $f_i(\by_{1:i-1})$ toward the implied $\bfxi_i^\top\by_{g_m(i)}$ and $d_i^2$ toward $\tau^2_{i}$, where $\bfSigma$ (and hence $\bfxi_i$ and $\tau^2_{i}$) are based on a parametric ``base'' covariance (e.g., Mat\'ern) that depends on parameters $\theta_p$. We suppress this dependence for now to ease notation and describe inference on $\theta_p$ in Section~\ref{sec:hyper}.

Specifically, for \eqref{eq:dprior}, we now assume $E(d_i^2) = \tau^2_{i}$ and $sd(d_i^2) = c_dE(d_i^2)$, where $c_d$ determines how much $d_i^2$ is shrunk toward $\tau^2_{i}$. Solving the inverse-Gamma moments $\bE(d_i^2) = \beta_i/ (\alpha_i - 1)$ and $sd(d_i^2) = \beta_i/(\alpha_i - 1)\sqrt{(\alpha_i - 2)}$ for $\alpha_i$ and $\beta_i$, we obtain
\begin{equation}
\alpha_i = (2 + 1/c_d^2), \qquad \beta_i = (1+1/c_d^2)\tau^2_{i}, \qquad i = 1, \ldots, N.
\label{eq:alphabetarepar}
\end{equation}

Conditional on the variance component $d_i^2$, we model each function $f_i$ using a GP:
\begin{equation}
    \label{eq:gpnew}
f_i | d_i \stackrel{ind.}{\sim} \GP(\bfxi_i^\top\by_{g_m(i)},d_i^2 K_i), \qquad i=1,\ldots,N,
\end{equation}
where $K_i(\cdot,\cdot) = C_i(\cdot,\cdot)/E(d_i^2)$. In contrast to \eqref{eq:gp}, where the regression function $f_i$'s are centered at zero, here we center our new prior distribution for $f_i$ at the conditional means $\bfxi_i^\top\by_{g_m(i)}$ obtained through the Vecchia approximation of $\bfSigma$. We write $C_i(\cdot, \cdot)$ as,
\begin{equation}
    \label{eq:kernelnew}
C_i(\by_{g_{m^\prime}(i)},\by_{g_{m^\prime}(i)}') = \sigma_0^2\by_{g_{m^\prime}(i)}^\top\bQ_i\by_{g_{m^\prime}(i)}^\prime + \sigma^2_i \, \rho(h_i(\by_{g_{m^\prime}(i)}^\top, \by^\prime_{g_{m^\prime}(i)})/\gamma), \qquad i=2,\ldots,N,
\end{equation}
where $\gamma$, $h_i$ and $\rho$ follow similar parameterization as in \eqref{eq:kernelreconstruct}. To obtain the sparsity parameter $m^\prime$ in \eqref{eq:sparsity}, we used $\varepsilon = 0.01$ for our numerical examples, which produced highly accurate inference and usually resulted in $m^\prime<10$. In contrast to \eqref{eq:kernelreconstruct}, we introduce a shrinkage parameter $\sigma^2_0$ in \eqref{eq:kernelnew} that determines the level of shrinkage towards a linear regression. 

It is important to note that we induce sparsity in the transport map through both assumptions \eqref{eq:vecchiai} and \eqref{eq:sparsity}. While $m^\prime$, the maximum number of nearest number for the nonlinear covariance kernel in \eqref{eq:kernelnew}, is determined from the data via $\bftheta_q$, we want the nearest-neighbor number $m$ in \eqref{eq:vecchiai} to be large enough for accurate Vecchia approximation of $\bfSigma$ but without jeopardizing computational efficiency, and so we use a fixed $m = 30$.

In addition, we also consider the same parameterization of $\sigma^2_i$ given by the former authors, i.e., 
%In contrast to \eqref{eq:gp} where the regression function $f_i$'s are centered around zero, here we center our new prior distribution for $f_i$ around the conditional means $\bfxi_i^\top\by_{g_m(i)}$ obtained through the Vecchia approximation in Section~\ref{sec:vecchia}. 
%Non-linearity is captured by $\sigma_i^2$ along with the covariance kernel $\rho_i(\cdot, \cdot)$ in \eqref{eq:kernelnew}. Here, we follow the parametrizations described by \citet{Katzfuss2021}. %We briefly outline the relevant details and refer interested readers to the original literature for detailed understanding. 
%Firstly, we assume %that the degree of nonlinearity $\sigma_i^2$ decays polynomially with length scale $\ell_i$ discussed in Section \ref{sec:maximin}, namely
$\sigma_i^2 = e^{\theta_{\sigma,1}} \ell_i^{\theta_{\sigma,2}}$, with $\ell_i$ as in Section~\ref{sec:maximin}, which allows the conditional distributions of $\by_{i:N}$ given $\by_{g_{m^\prime}(i)}$ to be increasingly Gaussian as $i$ increases, as a function of hyperparameters $\theta_{\sigma,1},\theta_{\sigma,2}$. This prior assumption is motivated by the behavior of stochastic processes with quasiquadratic loglikelihoods %that closely resemble Gaussian smoothness priors (with quadratic loglikelihoods) such as the Mat{\'e}rn model \citep{Whittle1954,whittle1963stochastic}. 
\citep{Katzfuss2021}. %argue that since the Bayesian transport map is formulated through the $\GP$s, the above parameterization can help us estimate the screening effect from the data. 

\begin{figure}
    \centering
    \includegraphics[width=0.5\linewidth]{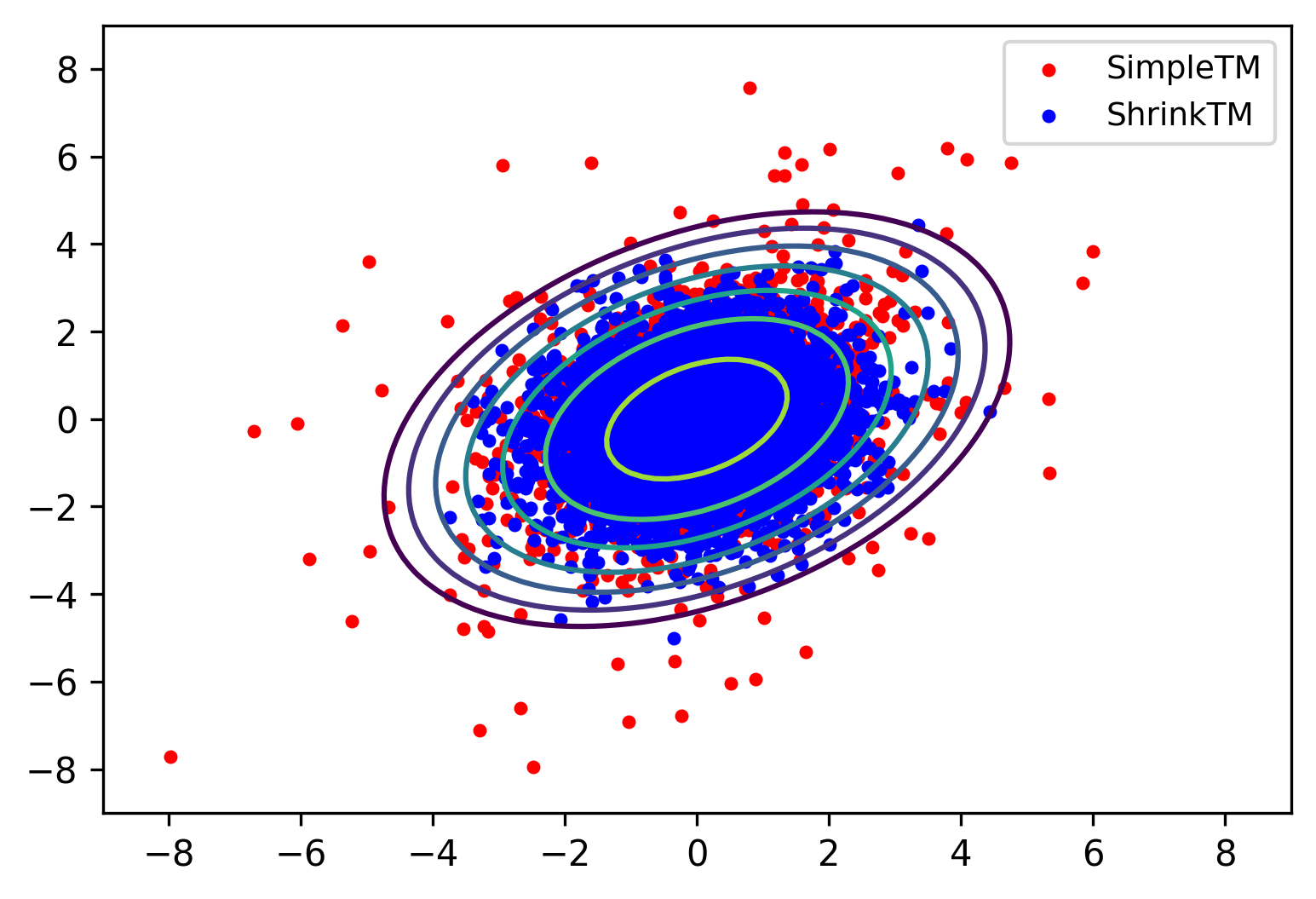}
    \caption{For two locations at distance 0.3 in the Gaussian simulation example in Figure~\ref{fig:lr900samples}, the true bivariate Gaussian distribution (contour lines), along with 5,000 samples each from the fitted existing Bayesian transport map (SimpleTM) outlined in Section \ref{sec:btmreview} (red) and our proposed ShrinkTM (blue), both trained on a single ($n=1$) sample.}
    \label{fig:contour_plot}
\end{figure}

Equations \eqref{eq:vecchiai}--\eqref{eq:kernelnew} imply a transport-map model that we call ShrinkTM, which can accurately learn the joint distribution even when trained with small number of samples. For example, Figure \ref{fig:contour_plot} shows that if we train both ShrinkTM and the existing transport-map approach (SimpleTM) outlined in Section \ref{sec:btmreview} on single simulated sample, ShrinkTM's parametric shrinkage allows it to capture the true distribution among two arbitrarily selected locations much more accurately than SimpleTM, which produces overly heavy tails.

\subsection{The posterior map \label{sec:invertiblemap}}

Now assume that we have observed $n$ independent training samples $\by^{(1)},\ldots,\by^{(n)}$ from the distribution in Section \ref{sec:btmreview} conditional on $\bf=(f_1,\ldots,f_N)$ and $\bd=(d_1,\ldots,d_N)$, such that 
$\by^{(j)} \stackrel{i.i.d.}{\sim} p(\by | \bf,\bd)$ with 
$\map(\by^{(j)}) \, | \, \bf,\bd \, \sim \normal_N(\bfzero,\bI_N)$, $j=1,\ldots,n$.
We assume that the samples are ordered according to the maximin ordering described in Section~\ref{sec:maximin} and combine the samples into an $n \times N$ data matrix $\bY$ whose $j$th row is given by $\by^{(j)}$. Then, for the regression in \eqref{eq:origtmapreform}, the responses $\by_i$ and the covariates $\bY_{g_m(i)}$ (and $\bY_{g_{m^\prime}(i)}$) are given by the $i$th and the $g_m(i)$ (and $g_{m^\prime}(i)$) columns of $\bY$, respectively. Below, let $\by^\star$ denote a new observation sampled from the same distribution, $\by^\star \sim p(\by | \bf,\bd)$, independently of $\bY$.

Based on the prior distribution for $\bf$ and $\bd$ in Section \ref{sec:fdpriornew}, we can now determine the posterior map $\pmap$ learned from the training data $\bY$, with $\bf$ and $\bd$ integrated out. This map is available in closed form and invertible:
\begin{proposition}
\label{prop:maps}
The transport map $\pmap$ from $\by^\star \sim p(\by|\bY)$ to 
$\bz^\star = \pmap(\by^\star) \sim \normal_N(\bfzero,\bI_N)$ is a triangular map with components
\begin{equation}
z_i^\star = \pmap_i(y_1^\star,\ldots,y_i^\star) = \Phi^{-1}\big( F_{2\tilde\alpha_i}\big( \hat d_i^{-1} (v_i(\by^\star_{1:i-1})+1)^{-1/2}(y_i^\star - \hat f_i(\by^\star_{1:i-1})) \big)\big), \quad i=1,\ldots,N, \label{eq:singlemap}
\end{equation}
where
$\tilde\alpha_i = \alpha_i + n/2$, 
$\tilde\beta_i = \beta_i + (\by_i - \bY_{g_m(i)}\bfxi_i)^\top \bG_i^{-1} (\by_i - \bY_{g_m(i)}\bfxi_i)/2$,
$\hat d_i^2=\tilde\beta_i/\tilde\alpha_i$,
$\bG_i = \bK_i + \bI_n$, $\bK_i= K_i(\bY_{g_{m^\prime}(i)},\bY_{g_{m^\prime}(i)}) =\big(K_i(\by_{g_{m^\prime}(i)}^{(j)},\by_{g_{m^\prime}(i)}^{(l)}) \big)_{j,l=1,\ldots,n}$,
\begin{align}
\hat f_i(\by^\star_{1:i-1}) & = \bG_i^{-1}\bY_{g_m(i)}\bfxi_i + K_i(\by^\star_{g_{m^\prime}(i)},\bY_{g_{m^\prime}(i)})\bG_i^{-1}\by_i, \label{eq:gppredmean}\\ 
v_i(\by^\star_{1:i-1}) & = K_i(\by^\star_{g_{m^\prime}(i)},\by^\star_{g_{m^\prime}(i)}) - K_i(\by^\star_{g_{m^\prime}(i)},\bY_{g_{m^\prime}(i)})\bG_i^{-1} K_i(\bY_{g_{m^\prime}(i)},\by^\star_{g_{m^\prime}(i)}), \label{eq:gppredvar}
\end{align}
for $i=2,\ldots,N$, $\hat f_1 = v_1 = 0$ for $i=1$, and $\Phi$ and $F_{\kappa}$ denote the cumulative distribution functions of the standard normal and the $t$ distribution with $\kappa$ degrees of freedom, respectively.
The inverse map $\pmap^{-1}$ can be evaluated at a given $\bz^\star$ by solving the nonlinear triangular system $\pmap(\by^\star) =\bz^\star$ for $\by^\star$, which can be expressed recursively as:
\begin{equation}
y_i^\star = \hat f_i(\by_{1:i-1}^\star) + F_{2\tilde\alpha_i}^{-1}(\Phi(z_i^\star))\, \hat d_i (v_i(\by_{1:i-1}^\star)+1)^{1/2}, \quad i=1,\ldots,N. \label{eq:invmap}
\end{equation}
\end{proposition}
Propositions \ref{prop:maps} and \ref{prop:lik} (see below) closely follow \citet{Katzfuss2021}, but due to our newly designed prior distributions in Section \ref{sec:fdpriornew}, we obtain posterior maps with different parameters, as detailed in the proofs in Appendix~\ref{app:proofs}.

Computation of $\pmap_i$ requires $\order(n^3 + m^\prime n^2)$ time for computing and decomposing the $n \times n$ matrix $\bG_i$ and $\order(m^3)$ time for computing the Vecchia approximation, for each $i=1,\ldots,N$. The $N$ components of $\pmap$ can be computed completely in parallel. 

\subsection{Hyperparameters \label{sec:hyper}}

\begin{algorithm}[t]
\caption{Estimation of the spatial transport map}
 \KwInput{Data matrix $\bY^{n \times N} = (\by^{(1)}, \ldots, \by^{(n)})^\top$.}
\KwResult{Trained transport map $\pmap_{\hat{\bftheta}}$.}
\begin{algorithmic}[1]
\STATE Order $y_1,\ldots,y_N$ in maximin ordering and compute scales $\ell_i$ and nearest-neighbor indices $c_i(1),\ldots,c_i(m_{\text{max}})$ (e.g., $m_{\text{max}}=30$) for each $i=1,\ldots,N$ (see Section \ref{sec:maximin})
\STATE Compute $\hat\bftheta = \arg\max_{\bftheta} \log p_{\bftheta}(\bY)$ via stochastic gradient ascent, where $ p_{\bftheta}(\bY)$ is given in \eqref{eq:intlik}.
\STATE Use fitted map to generate new samples or to find score of an observation. A new sample can be generated using $\by^\star=\pmap_{\hat{\bftheta}}^{-1}(\bz^\star)$ using \eqref{eq:invmap} based on $\bz^\star \sim \normal_N(\bfzero,\bI_N)$. Score can be calculated using \eqref{eq:tstar}.
\end{algorithmic}
\label{alg:inf}
\end{algorithm}

In the prior distributions of $f_i$ and $d_i$ in Section \ref{sec:fdpriornew}, $\alpha_i$, $\beta_i$, $m$, $K_i$ depend on a vector $\bftheta = (\bftheta_p,c_d,\theta_{\sigma,1},\theta_{\sigma,2},\theta_q)$ of hyperparameters, where $\bftheta_p$ are parameters determining the parametric covariance $\bfSigma$, $\theta_q$ determines decay and $m^\prime$, and $c_d,\theta_{\sigma,1},\theta_{\sigma,2}$ determine the strength of shrinkage toward (a Vecchia approximation with conditioning-set size $m$ of) $\normal(\bfzero,\bfSigma)$. We can write in closed form the integrated likelihood $p_{\bftheta}(\bY)$, where $\bf$ and $\bd$ have been integrated out. 

\begin{proposition}
\label{prop:lik}
The integrated likelihood is
\begin{equation}
    p_{\bftheta}(\bY) \textstyle \propto \prod_{i=1}^N \big( \, |\bG_i|^{-1/2} \times ({\beta_i^{\alpha_i}}/{\tilde\beta_i^{\tilde\alpha_i}}) \times {\Gamma(\tilde\alpha_i)}/{\Gamma(\alpha_i)} \, \big),
    \label{eq:intlik}
\end{equation}
where $\Gamma(\cdot)$ denotes the gamma function, and $\tilde\alpha_i$, $\tilde\beta_i$, $\bG_i$ are defined in Proposition \ref{prop:maps}.
\end{proposition}

For our numerical results, we have followed the strategy of \citet{Katzfuss2021} and have employed the empirical Bayesian approach, because it is fast and preserves the closed-form map properties in Section \ref{sec:invertiblemap}. Once optimized, the resulting posterior transport map can be used to draw inference on the training data by drawing new samples or calculating the score function (as in Algorithm \ref{alg:inf}).

\begin{figure}[!htbp]
    \centering
    \includegraphics[width=\linewidth]{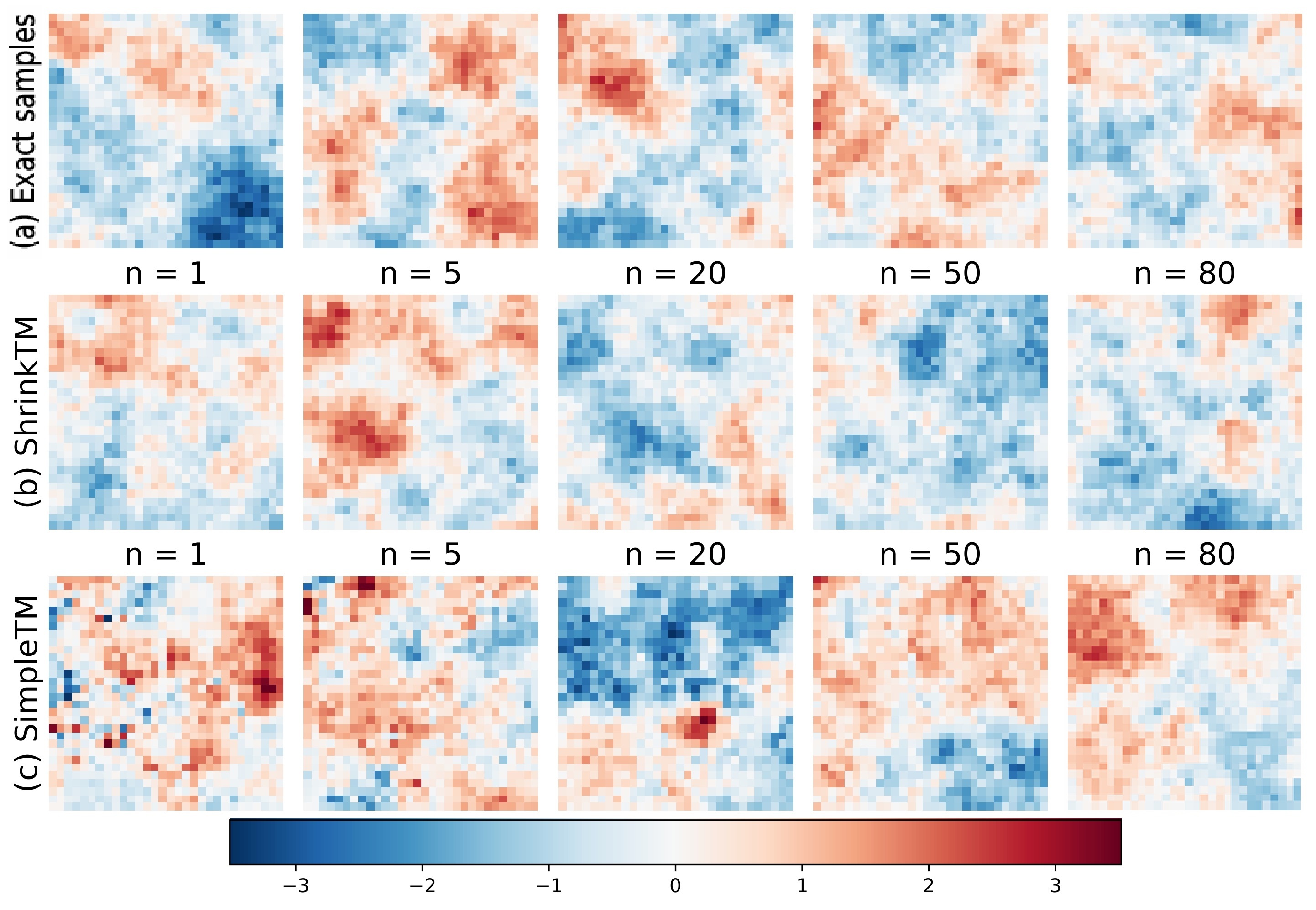}
    \caption{Exact samples from LR900 (top) and (independent) samples from ShrinkTM (middle) and SimpleTM (bottom) trained on $n$ samples from LR900}
    \label{fig:lr900samples}
\end{figure}

%%%%%%%%%%%%%%%
\section{Numerical results\label{sec:numerical}}

\subsection{Simulation experiments \label{sec:simulation}}
For our numerical experiments, we consider two separate simulation scenarios on a regular spatial grid of $N = 30\times 30 = 900$ locations in a unit square that were previously considered by \citet{Katzfuss2021}.

\textbf{LR900:} $\by \sim \normal_{900}(0, \bV)$, where $\bV$  is an exponential covariance with unit variance and range parameter 0.3 (as in Figure \ref{fig:lr900samples}(a)). This data generating mechanism can also be described as a linear map $f_L(\by_{1:i-1}) = \sum_{k=1}^{i-1}b_{i,k}y_{c_i(k)}$, where the $b_{i,k}y_{c_i(k)}$ are coefficients of the conditional means (as in \eqref{eq:vecchiai}) computed from the covariance matrix $\bV$.

\textbf{NR900:} A sine function is added to the map components $f_i^\text{NL}(\by_{1:i-1}) = f_i^\text{L}(\by_{1:i-1}) + 2 \sin(4(b_{i,1} y_{c_i(1)}+b_{i,2} y_{c_i(2)}))$, where the $b_{i,k}y_{c_i(k)}$s are coefficients of the conditional means computed from $\bfSigma$ used in LR900 (as in Figure \ref{fig:nr900samples}(a)).

We consider both the existing BTM approach of \citet{Katzfuss2021} (which we name \textbf{SimpleTM}) and our new version (i.e., \textbf{ShrinkTM}). For each of the two simulation settings, we randomly sample $n$ training samples and train both SimpleTM and ShrinkTM. We vary the value of $n$ and repeat this experiment 10 times for each value of $n$. For all the training tasks, we use default initial values $\bftheta = (\bftheta_p,c_d,\theta_{\sigma,1},\theta_{\sigma,2},\theta_q) = (2.0, 0.0, 0.0, 0.0, -1.0)$.
We perform the training of both SimpleTM and ShrinkTM in Python using the stochastic-gradient-descent-based Adam optimization algorithm from the \texttt{PyTorch} 2.5.1 library. We set the initial learning rate of the Adam optimizer to be 0.01 and adjust it using the cosine annealing rule in the \texttt{PyTorch} library. The remaining tuning parameters of the optimizer are fixed at their default values. We also compare to a \textbf{MatCov} approach that refers to a Gaussian process with an isotropic Mat\'ern covariance (which is also used as the base covariance in ShrinkTM), with the three hyperparameters inferred via maximum likelihood estimation.
\begin{figure}
    \centering
    \includegraphics[width=\linewidth]{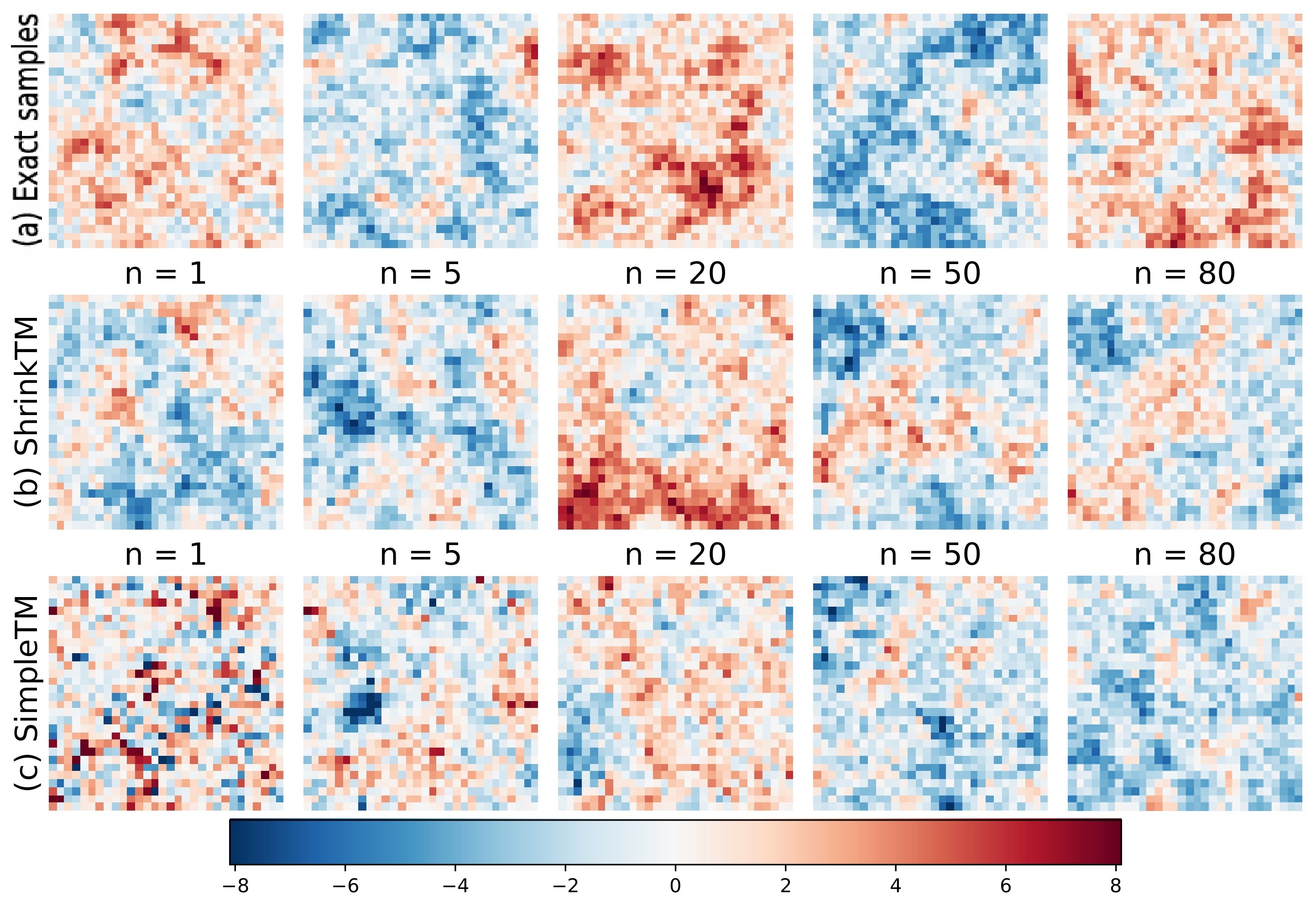}
    \caption{Exact samples from NR900 (top) and (independent) samples from ShrinkTM (middle) and SimpleTM (bottom) trained on $n$ samples from NR900}
    \label{fig:nr900samples}
\end{figure}

Figures \ref{fig:lr900samples} and \ref{fig:nr900samples} show samples generated from SimpleTM and ShrinkTM for different training sizes for LR900 and NR900, respectively. We can see that even for $n=1$ training samples, ShrinkTM excels in capturing the long- and short-range dependencies due to better regularization through the conditional means and variances, whereas SimpleTM struggles. We can furthermore see in Figure~\ref{fig:simshrinkageresults}(a) that for LR900, the shrinkage factor $c_d$ lies between 0.1 and 0.25 for $n=1$ with the average value of $c_d$ around 0.15, indicating the efficiency of ShrinkTM. As $c_d < 1$ means $sd(d_i) < E(d_i)$, this suggests a moderately strong shrinkage of the variance components $d_i$ toward the corresponding values $\tau_i(\hat\theta_p)$'s (as discussed in Section~\ref{sec:fdpriornew}). The average of $c_d$ decreases for large training sample sizes, indicating a stronger shrinkage. On the other hand, since NR900 has a modified dependence structure relative to an exponential covariance, we can see that $c_d$ varies for different sample sizes and automatically decides the amount of shrinkage from the training data.

\begin{figure}
    \centering
    \includegraphics[width=\linewidth]{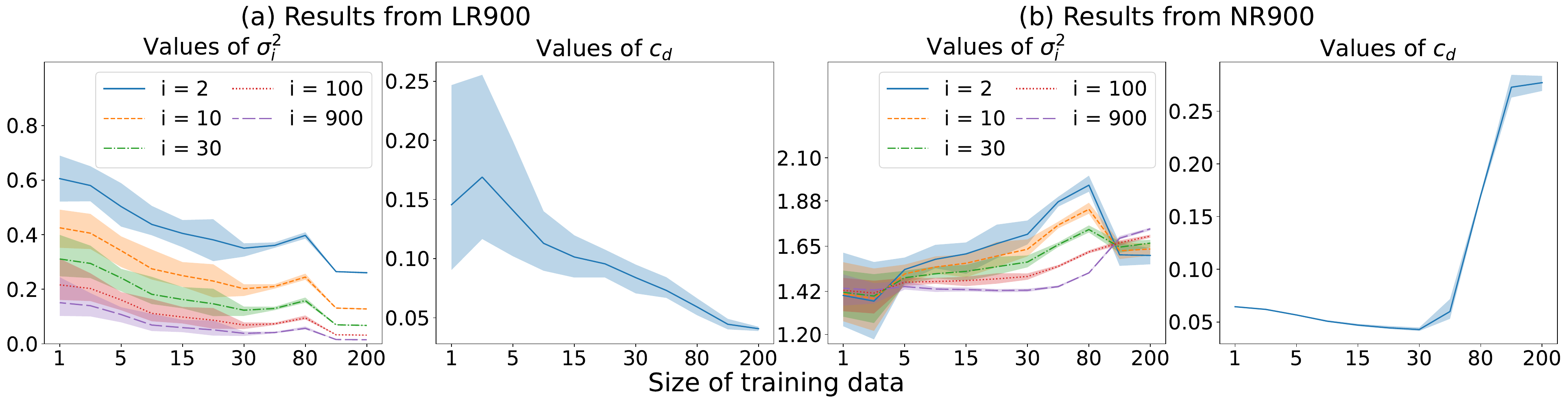}
    \caption{Results of several parameters from simulation experiments described in Section~\ref{sec:simulation}. First two columns represent values of $\sigma_i^2$ and $c_d$ for the \textbf{LR900} experiment, while the second two columns represent values of the $\sigma_i^2$ and $c_d$ for \textbf{NR900} experiment.}
    \label{fig:simshrinkageresults}
\end{figure}

\begin{figure}
    \centering
    \includegraphics[width=\linewidth]{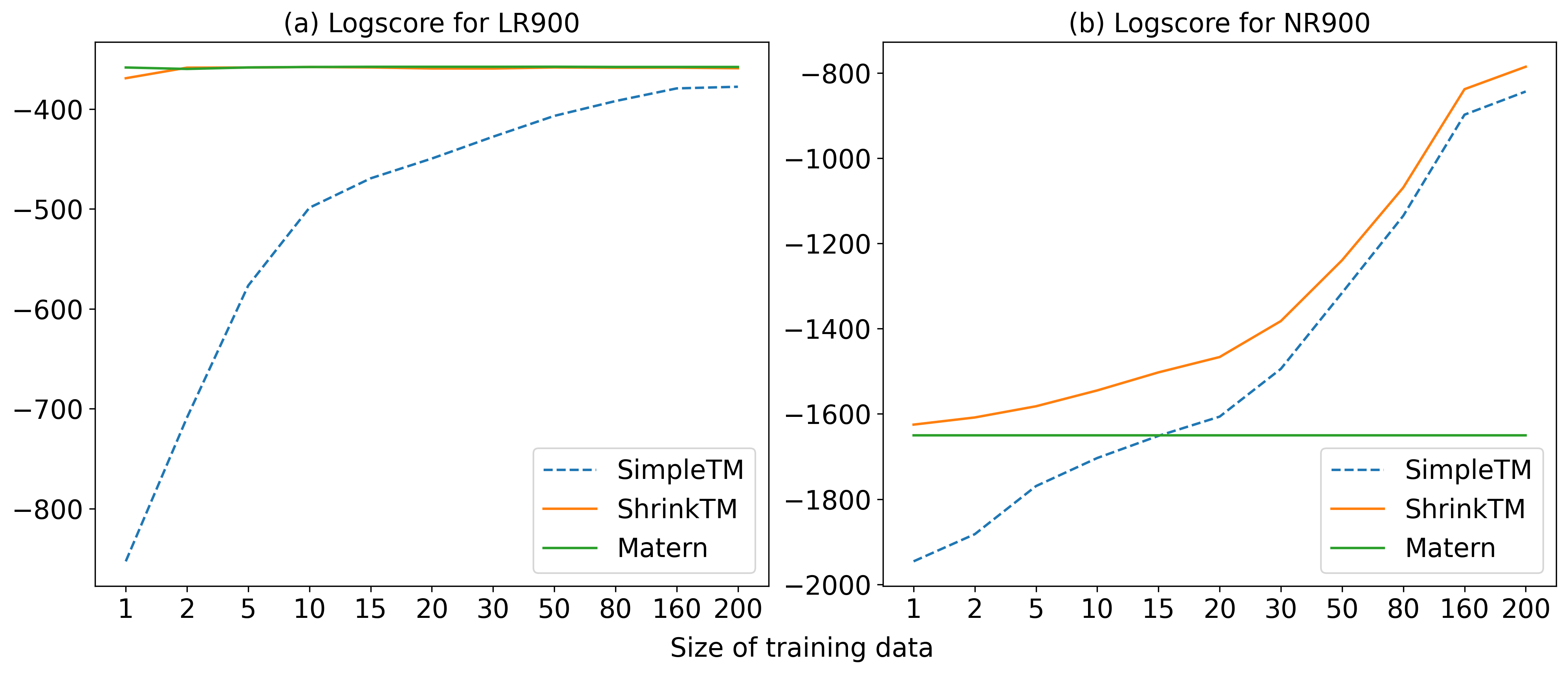}
    \caption{Logarithmic scores for \textbf{LR900} (left) and \textbf{NR900} (right) for varying sample sizes.}
    \label{fig:simdatalogscores}
\end{figure}

To further assess the accuracy of these two methods, we generate test samples from the two simulation scenarios, and we compare average logarithmic score (higher is better) of SimpleTM, ShrinkTM, and MatCov in Figure~\ref{fig:simdatalogscores}. The log-score \citep[e.g.,][]{Gneiting2014} approximates the negative Kullback-Leibler divergence between the true and approximated distribution up to an additive constant. For LR900, the log-score of ShrinkTM is much higher than that of SimpleTM for very small numbers of training samples $n$ (as seen in Figure~\ref{fig:simdatalogscores}(a)), and despite its flexible nonparametric structure, it matches the accuracy of the correctly specified parametric ``base" MatCov (which only needs to estimate three parameters) when trained on only $n=2$ training samples. Even when data is being generated from NR900, ShrinkTM consistently outperforms SimpleTM (as seen in Figure~\ref{fig:simdatalogscores}(b)). Average training time of ShrinkTM was also comparable with that of SimpleTM: On the uniform grid of $N=900$ points, on average, ShrinkTM took 49 seconds and SimpleTM took 7 seconds to train on $n=1$ replicate, but ShrinkTM and SimpleTM took 15 and 13.5 minutes, respectively, when training on $n=200$ samples. This is in line with the time complexities of $\order(Nm^3)$ for the Vecchia approximation (only in ShrinkTM) and $\order(N(n^3 + m^\prime n^2))$ for computing and decomposing the $\bG_i$ (for both ShrinkTM and SimpleTM), where the latter increases and starts to dominate for increasing $n$.

\subsection{Climate data application\label{sec:climateapplication}}

\begin{figure}
    \centering
    \includegraphics[width=\linewidth]{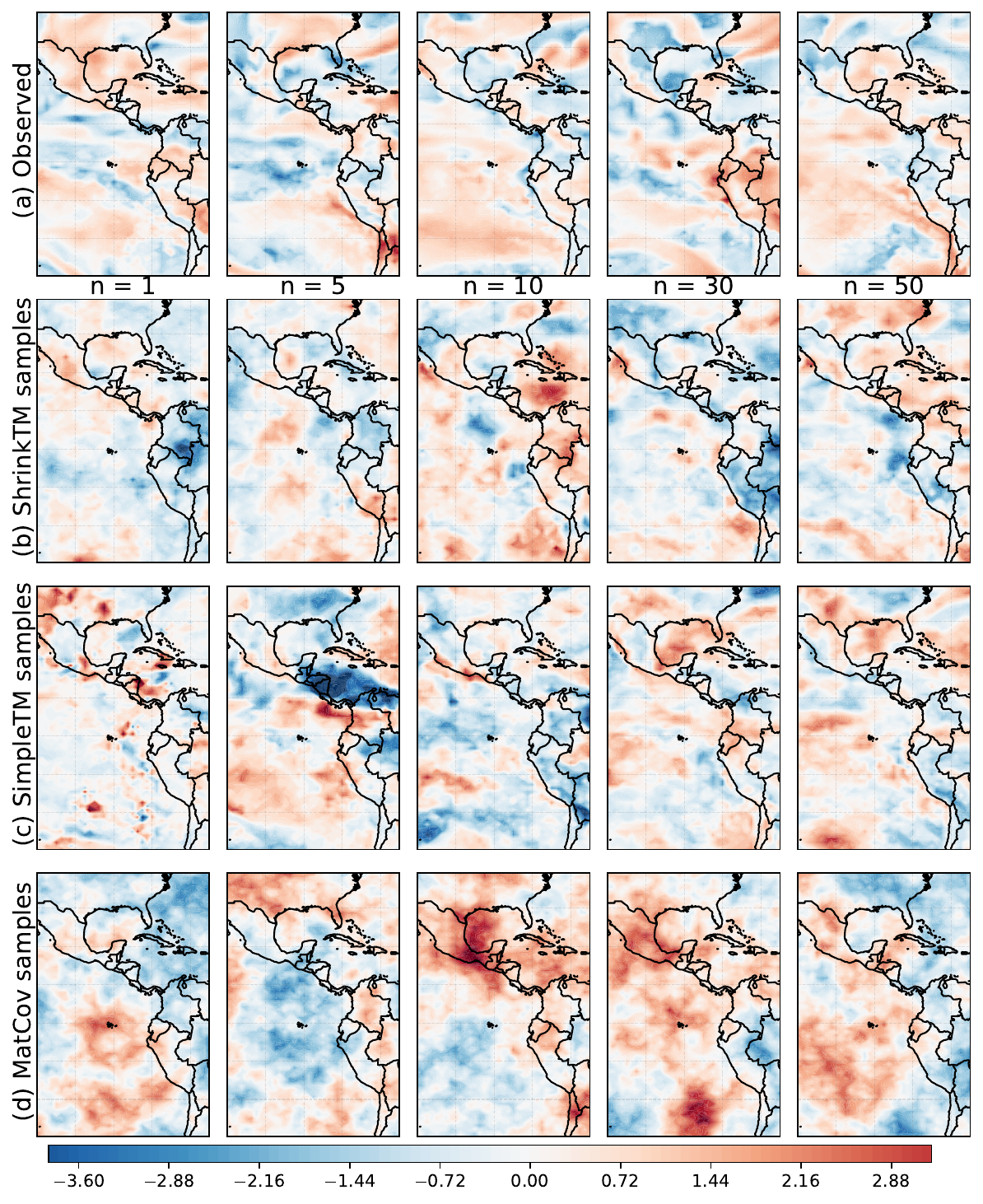}
    \caption{Observed and fitted samples for the climate model experiment outlined in Section~\ref{sec:simulation}. Row (a) represents 5 observed samples from the climate model run described in Section~\ref{sec:climateapplication}. Row (b), (c) and (d) represent samples from ShrinkTM, SimpleTM and MatCov for varying training ensemble size, denoted by $n$.}
    \label{fig:americas_samples}
\end{figure}

\begin{figure}
    \centering
    \includegraphics[width=\linewidth]{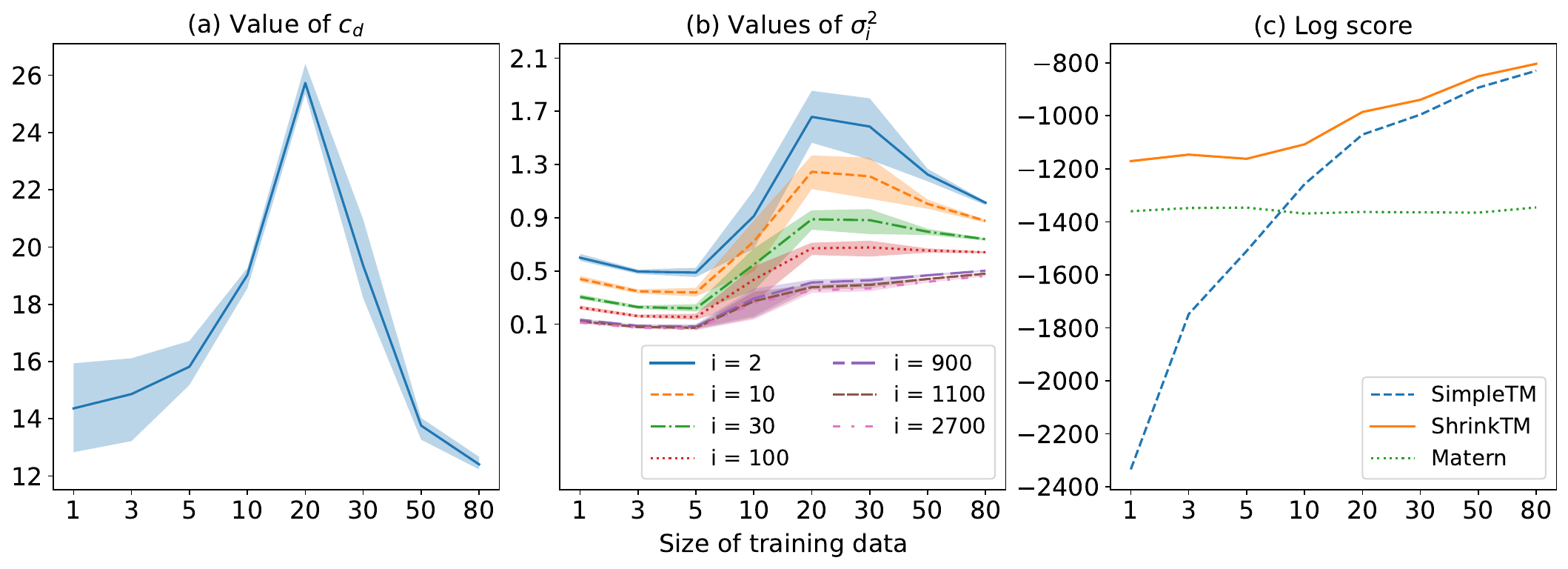}
    \caption{Shrinkage metrics and logarithmic scores at varying training ensembles of log-precipitation data outlined in Section~\ref{sec:climateapplication}.}
    \label{fig:americas_numerics}
\end{figure}
We consider log-transformed total precipitation rate (in m/s) on a roughly $1^\circ$ longitude-latitude global grid of size N = 37 × 74 = 2738 in the middle of the Northern summer(July 1) in 98 consecutive years, starting in the year 402, from the Community Earth System Model (CESM) Large Ensemble Project \citep{Kay2015}. This dataset was also analyzed in \citet{Katzfuss2021}. CESM belongs to a broader class of climate models, which are large sets of computer code describing the behavior of the Earth system (e.g., the atmosphere) via systems of differential equations. Enormous computational power is required to produce even a single ensemble of these models on fine latitude-longitude grids from these computer codes. Due to the massive cost of obtaining a single sample from these climate models, it is important to build a stochastic generative model based on few training samples from the climate model, in order to produce relevant summaries and even draw more samples at much cheaper cost.

\begin{figure}[!htbp]
    \centering
    \includegraphics[width=\linewidth]{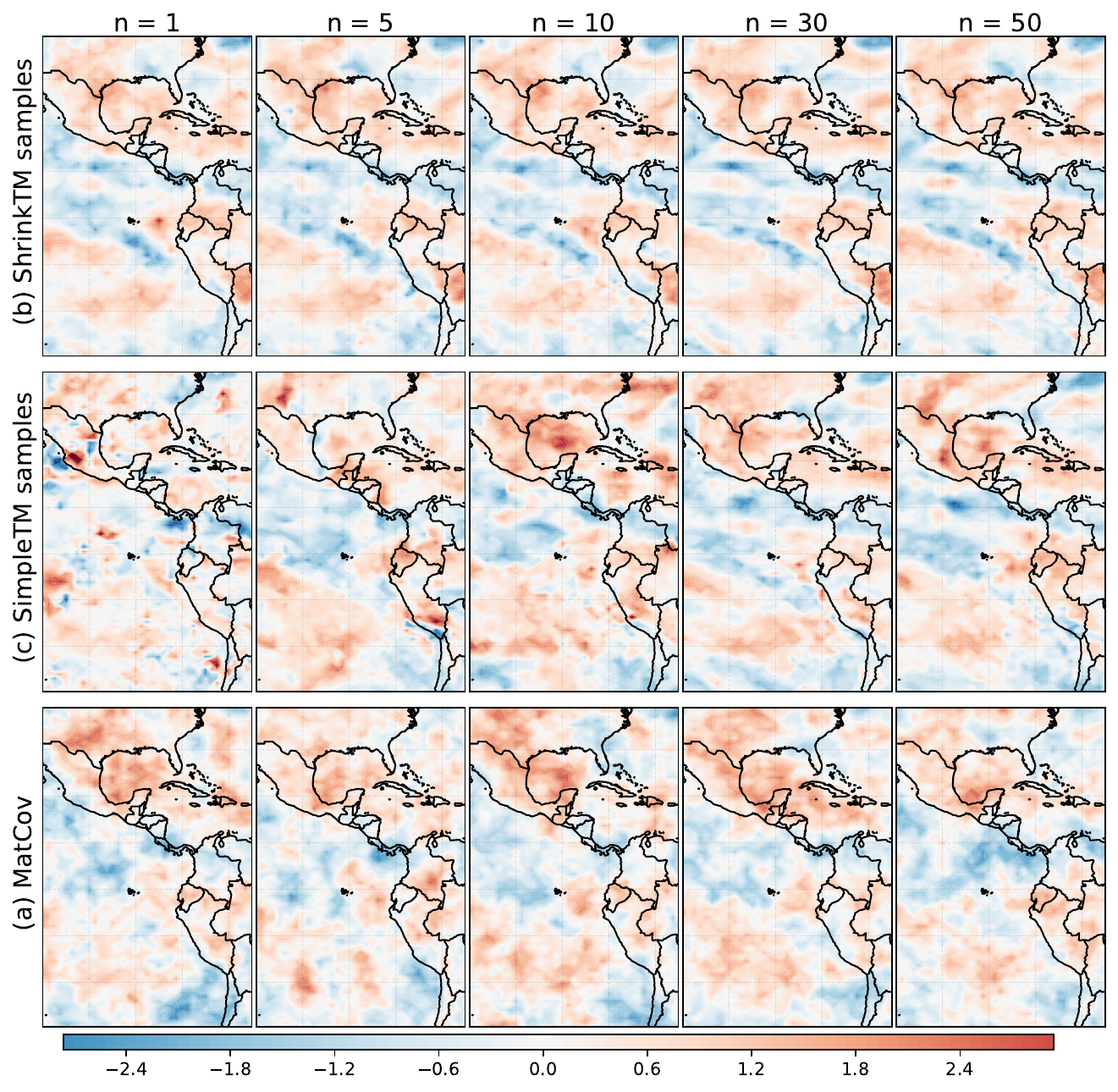}
    \caption{Partially observed samples generated from MatCov, ShrinkTM and SimpleTM for varying training ensembles. For each of these samples, the output at first 100 ordered locations have been fixed at the values from the figure in the first column of Figure~\ref{fig:americas_samples}(a).}
    \label{fig:americas_partsamples}
\end{figure}
We consider a subsection of the western hemisphere (39.1$^\circ$N to 29.6$^\circ$S and 110$^\circ$W to 65$^\circ$W) containing parts of land including North, Central and South America, and a subsection of Atlantic and Pacific ocean. We obtain precipitation anomalies by standardizing the data at each grid location to mean zero and variance 1 (as shown in Figure~\ref{fig:americas_samples}(a)), and used them as training ensembles.

\begin{figure}
    \centering
    \includegraphics[width=0.5\linewidth]{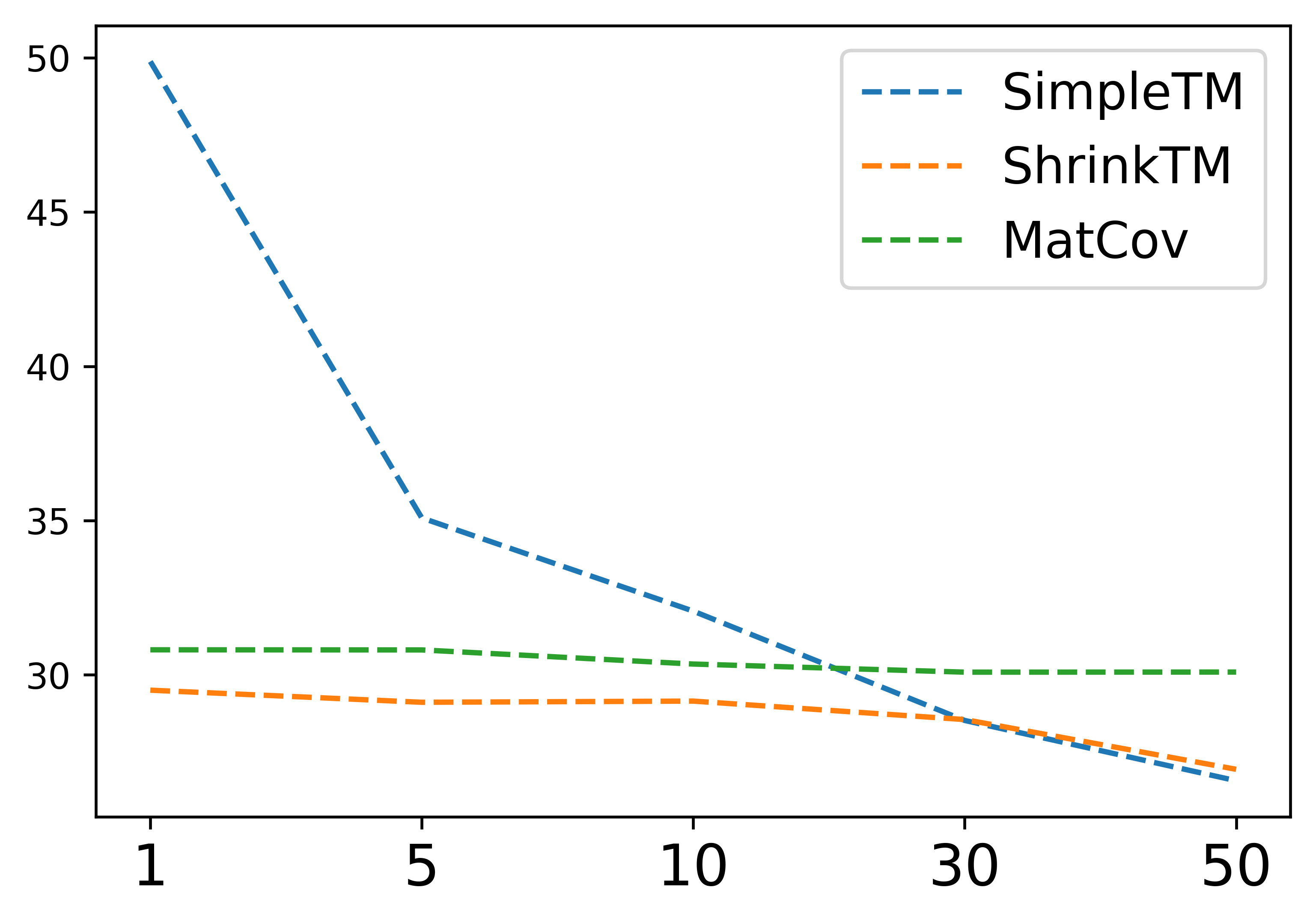}
    \caption{In the partially observed setting of Figure~\ref{fig:americas_partsamples}, root mean square error at held-out locations for conditional samples from three generative models: SimpleTM, ShrinkTM, and GP with Mat\`ern kernel.}
    \label{fig:climate_mse}
\end{figure}

Figure \ref{fig:americas_samples} displays a set of samples obtained from SimpleTM, ShrinkTM, and MatCov for several training ensemble sizes. It is evident that with a low number of training samples (i.e., $n = 1$), ShrinkTM is able to capture non-Gaussian features and long-range characteristics, while SimpleTM requires 10 or more training samples to exhibit such features. When trained on $n=1, 5, 10, 30, 50$ training samples, ShrinkTM selects $m'=3, 3, 4, 16, 20$ nearest neighbors to learn non-Gaussian distribution, respectively. On average, training of ShrinkTM takes 39 seconds, in comparison to 3 seconds taken by SimpleTM for $n=1$. The average training time increases to 9, 18, 82 seconds for SimpleTM, and 46, 53, 123 seconds for ShrinkTM when training on $n=10, 20, 50$ samples, respectively. Figure~\ref{fig:americas_numerics}(a) shows that for small $n$, the shrinkage factor $c_d$ has an extremely high value, suggesting non-Gaussianity in the samples. The level of non-Gaussianity is better learned with increase in training-ensemble size.

We quantitatively compare the performance of ShrinkTM and SimpleTM using average log-score. Figure~\ref{fig:americas_numerics}(c) shows the average log-score for these three methods with varying number of training ensembles. ShrinkTM outperforms both SimpleTM and MatCov for all training-sample sizes, including providing higher accuracy than the ``base'' MatCov approach even for a single training sample.

We consider conditional simulation to sample precipitation anomalies when a sample is partially observed. To do this, we consider the precipitation anomalies in the top-left panel of Figure~\ref{fig:americas_samples} (which were not part of the training data) and assume that we have observed only the first $100$ ordered locations of this test field. Figure~\ref{fig:americas_partsamples} displays the conditional samples for different methods for varying training ensembles conditional on this partially observed field. It can be seen (Figure~\ref{fig:americas_partsamples}(b)) that with only a few number of training ensembles, ShrinkTM is able to capture most of the long range dependencies that are seen in Figure~\ref{fig:americas_samples} (topleft panel) as well. SimpleTM needs around 30 samples to capture these dependencies, while MatCov is not able to capture them at all.
Figure~\ref{fig:climate_mse} shows root mean square error (RMSE) of predicting the held-out locations when repeating this experiment 10 times.

\section{Conclusion \label{sec:conclusion}}

We proposed an improved version of the BTM \citep{Katzfuss2021} by centering the mean and variance functions of its map components around the conditional means and variances obtained from a Gaussian process (GP) with a parametric covariance, whose parameters are estimated from the training data. We also employ a Vecchia approximation in the calculation of conditional means and variances of the GP. This new BTM version, which we call ShrinkTM, preserves  the scalability and flexibility of BTM. A notable advantage of ShrinkTM is that it can cheaply produce realistic samples even when the transport map is trained with only a single training ensemble, and it is thus especially useful when few or even only a single training sample is available. ShrinkTM is nonparametric and largely automated, in that it does not require specifying the form or extent of non-Gaussianity or nonstationarity. We have demonstrated these crucial advantages using simulated datasets and log-precipitation datasets obtained from the CESM climate model.  A Python implementation of our methods, along with code to reproduce all results, is available at \url{https://github.com/katzfuss-group/batram/tree/ShrinkTM}. 
Our work can be further extended to multivariate transport maps (e.g., \citet{Wiemann2023BayesianFields}). Moreover, we can vary the hyperparameters of the Gaussian process over a covariate space and thus produce a covariate-dependent shrinkage of the heteroscedastic transport map \citep{drennan2024covtmprep}, which can be applicable to a wider range of applications. We plan to tackle these problems in future work.

%%%%%%%%%%%%%%%%%%%%%%%%%%%%%%%%%%%%%%%%%%%%%%%%%%%%%%%
\footnotesize
\appendix
\section*{Acknowledgments}

This work was supported by NASA's Advanced Information Systems Technology Program (AIST-21). 
MK's research was also partially supported by National Science Foundation (NSF) Grants DMS--1953005 and DMS--2433548 and by the Office of the Vice Chancellor for Research and Graduate Education at the University of Wisconsin--Madison with funding from the Wisconsin Alumni Research Foundation. We would like to thank Daniel Drennan for helpful comments on the \texttt{PyTorch} implementation.

%%%%%%%%%%%%%%%%%%%%%%%%%%%%%%%%%%%%%%%%%%%%%%%%%%%%%%%

\section{Proofs \label{app:proofs}}
Similar to the propositions, the proofs also closely follow \citet{Katzfuss2021}.
\begin{proof}[Proof of Proposition \ref{prop:maps}]
Combining \eqref{eq:origtmapreform} with the conditional independence of $\by^{(1)},\ldots,\by^{(n)}$, we have
\begin{equation}
    \label{eq:reg}
\textstyle p(\bY|\bf,\bd) = \prod_{i=1}^N \prod_{j=1}^n \normal(y_i^{(j)}|f_i(\by_{1:i-1}^{(j)}),d_i^2) = \prod_{i=1}^N \normal_n(\by_i|\bf_i,d_i^2\bI_n),
\end{equation}
where $\bf_i = f_i(\bY_{1:i-1}) = \big(f_i(\by_{1:i-1}^{(1)}),\ldots,f_i(\by_{1:i-1}^{(n)}) \big)^\top$ is distributed as $\bf_i|d_i,\bY_{1:i-1} \sim \normal(\bY_{g_m(i)}\xi^\top_i,d_i^2\bK_i)$. 
Combined with \eqref{eq:dprior}, we see that $\bf_i, d_i$ (conditional on $\bY_{1:i-1}$) jointly follow a (multivariate) normal-inverse-gamma (NIG) distribution, independently for each $\bf_i,d_i$. 
Given the data $\bY$ as in \eqref{eq:reg}, well-known conjugacy results imply that the posterior of $\bF = (\bf_1,\ldots,\bf_N)$ and $\bd=(d_1,\ldots,d_N)$ also consists of independent NIG distributions:
\begin{equation}
\textstyle p(\bF,\bd|\bY)  %= p(\bF|\bY,\bd) p(\bd|\bY)
\propto \prod_{i=1}^N p(\by_i|\bf_i,d_i) p(\bf_i|d_i,\bY_{1:i-1}) \, p(d_i) \propto \prod_{i=1}^N \normal(\bf_i|\hat\bf_i,d_i^2 \tilde\bK_i) \, \mathcal{IG}(d_i^2|\tilde\alpha_i,\tilde\beta_i), \label{eq:nigpost}
\end{equation}
where $\tilde\bK_i = \bK_i - \bK_i \bG_i^{-1}\bK_i$ and $\hat\bf_i = \bG_i^{-1}\bY_{1:i-1}\Lambda_i^\top + \bK_i \bG_i^{-1} \by_i$.

We have
$p(\by^\star|\bf, \bd) = \prod_{i=1}^N \normal(y_i^\star|f_i(\by_{1:i-1}^*),d_i^2)$ using \eqref{eq:origtmapreform}. Combining this with \eqref{eq:nigpost} and the conditional-independence assumption in \eqref{eq:gp}, the posterior predictive distribution can be shown to be
\[
p(\by^\star|\bY) = \prod_{i=1}^N \int p(y_i^\star|\by_{1:i-1}^*,\bY,d_i) p(d_i|\bY) d d_i,
\]
where basic GP regression implies
\begin{equation}
\label{eq:condpred}
y_i^\star|\by^\star_{1:i-1},\bY,d_i \sim \normal\big(\hat f_i(\by^\star_{1:i-1}), d_i^2( v_i(\by^\star_{1:i-1})+1 )\big), \qquad i=1,\ldots,N.
\end{equation}
Combining this with $d_i^2 | \bY \sim \mathcal{IG}(\tilde\alpha_i,\tilde\beta_i)$ from \eqref{eq:nigpost},
we obtain the posterior predictive distribution as a product of $t$ densities,
\begin{equation}
    \label{eq:tstar}
    \textstyle p(\by^\star|\bY) = \prod_{i=1}^N t_{2\tilde\alpha_i}\big(y_i^\star\big|\hat f_i(\by^\star_{1:i-1}), \hat d_i^2( v_i(\by^\star_{1:i-1})+1 )\big),
\end{equation}
where our notation is such that $w \sim t_{\kappa}(\mu,\sigma^2)$ implies that $(w - \mu)/\sigma$ follows a ``standard'' $t$ with $\kappa$ degrees of freedom.
Hence, $\hat d_i^{-1} (v_i(\by^\star_{1:i-1})+1)^{-1/2}(y_i - \hat f_i(\by^\star_{1:i-1}))$ follows a $t_{2\tilde\alpha_i}$ distribution.
Using the fact that we can map from a distribution to the standard uniform using its cumulative distribution, the transformation
$\bz^\star = \pmap(\by^\star) \sim \normal_N(\bfzero,\bI_N)$ to a standard normal can be described using a triangular map with components
\begin{equation}
\label{eq:singlemap2}
z_i^\star = \pmap_i(y_1^\star,\ldots,y_i^\star) = \Phi^{-1}\big( F_{2\tilde\alpha_i}\big( \hat d_i^{-1} (v_i(\by^\star_{1:i-1})+1)^{-1/2}(y_i^\star - \hat f_i(\by^\star_{1:i-1})) \big)\big).
\end{equation}
The solution $\by^\star$ to the nonlinear triangular system $\pmap(\by^\star) =\bz^\star$ is found recursively by solving \eqref{eq:singlemap2} for $y_i^\star$:
\begin{equation}
y_i^\star = F_{2\tilde\alpha_i}^{-1}(\Phi(z_i^\star))\, \hat d_i (v_i(\by_{1:i-1}^\star)+1)^{1/2} + \hat f_i(\by_{1:i-1}^\star).
\end{equation}
\end{proof}

\begin{proof}[Proof of Proposition \ref{prop:lik}]
From \eqref{eq:gpnew}, we have that $\bf_i | d_i \stackrel{ind.}{\sim} \normal_n(\bfxi_i^\top\by_{g_m(i)},d_i^2\bK_i)$; together with \eqref{eq:reg}, this implies that $\by_i|d_i,\bY_{1:i-1} \stackrel{ind.}{\sim}  \normal_n(\bY_{g_m(i)}\bfxi_i,d_i^2\bG_i)$. Combining this with \eqref{eq:dprior}, it is well known that $\by_i | \bY_{1:i-1} \stackrel{ind.}{\sim} t_{2\alpha_i}(\bY_{g_m(i)}\bfxi_i,\frac{\beta_i}{\alpha_i}\bG_i)$, where we define a multivariate $t$ distribution such that $\bw \sim t_{\kappa}(\bfmu,\bfSigma)$ implies that the entries of $\bfSigma^{-1/2}(\bw - \bfmu)$ are i.i.d.\ standard $t$ with $\kappa$ degrees of freedom. 
Plugging in the $t$ densities and simplifying using $\tilde\alpha_i = \alpha_i + n/2$, 
$\tilde\beta_i = \beta_i + (\by_i - \bY_{g_m(i)}\bfxi_i)^\top \bG_i^{-1} (\by_i - \bY_{g_m(i)}\bfxi_i)/2$, we can obtain
\begin{align}
  p(\bY) 
  & = \textstyle \prod_{i=1}^N t_{2\alpha_i}(\by_i |\bfzero,\frac{\beta_i}{\alpha_i}\bG_i)\\
  & \propto \textstyle \prod_{i=1}^N \Gamma(\tilde\alpha_i) \big( \Gamma(\alpha_i) (\alpha_i \beta_i/\alpha_i)^{n/2} |\bG_i|^{1/2}\big)^{-1} \big( 1+ \alpha_i/(\beta_i 2 \alpha_i) \by_i^\top \bG_i^{-1}\by_i \big)^{-\tilde\alpha_i}\\
  & \propto \textstyle\prod_{i=1}^N \big( \, |\bG_i|^{-1/2} \times ({\beta_i^{\alpha_i}}/{\tilde\beta_i^{\tilde\alpha_i}}) \times {\Gamma(\tilde\alpha_i)}/{\Gamma(\alpha_i)} \, \big),
    % \\ & \textstyle \propto \prod_{i=1}^N |\bG_i|^{-1/2} \times {\beta_i^{\alpha_i}}/{\tilde\beta_i^{\tilde\alpha_i}} \times {\Gamma(\tilde\alpha_i)}/{\Gamma(\alpha_i)}
\end{align}
where $\Gamma(\cdot)$ denotes the gamma function.
\end{proof}

%%%%%%%%%%%%%%%%%%%%%%%%%%%%%%%%%%%%%%%%%%%%%%%%%%%%%%%
\footnotesize
\appendix
\section*{Acknowledgments}

This work was supported by NASA's Advanced Information Systems Technology Program (AIST-21). 
MK's research was also partially supported by National Science Foundation (NSF) Grants DMS--1953005 and DMS--2433548 and by the Office of the Vice Chancellor for Research and Graduate Education at the University of Wisconsin--Madison with funding from the Wisconsin Alumni Research Foundation. We would like to thank Daniel Drennan for helpful comments on the \texttt{PyTorch} implementation.

%%%%%%%%%%%%%%
\bibliographystyle{apalike}
\bibliography{main}

\begin{thebibliography}{}

\bibitem[Arjovsky and Bottou, 2017]{Arjovsky2017}
Arjovsky, M. and Bottou, L. (2017).
\newblock {Towards principled methods for training generative adversarial networks}.
\newblock In {\em International Conference on Learning Representations}.

\bibitem[Banerjee et~al., 2004]{Banerjee2004}
Banerjee, S., Carlin, B.~P., and Gelfand, A.~E. (2004).
\newblock {\em {Hierarchical Modeling and Analysis for Spatial Data}}.
\newblock Chapman {\&} Hall.

\bibitem[Carlier et~al., 2009]{Carlier2009}
Carlier, G., Galichon, A., and Santambrogio, F. (2009).
\newblock {From Knothe's transport to Brenier's map and a continuation method for optimal transport}.
\newblock {\em SIAM Journal on Mathematical Analysis}, 41(6):2554--2576.

\bibitem[Castruccio et~al., 2014]{Castruccio2014}
Castruccio, S., McInerney, D.~J., Stein, M.~L., Crouch, F.~L., Jacob, R.~L., and Moyer, E.~J. (2014).
\newblock {Statistical emulation of climate model projections based on precomputed GCM runs}.
\newblock {\em Journal of Climate}, 27(5):1829--1844.

\bibitem[Cressie, 1993]{Cressie1993}
Cressie, N. (1993).
\newblock {\em {Statistics for Spatial Data, revised edition}}.
\newblock John Wiley {\&} Sons, New York, NY.

\bibitem[Datta et~al., 2016]{Datta2016}
Datta, A., Banerjee, S., Finley, A.~O., and Gelfand, A.~E. (2016).
\newblock {Hierarchical nearest-neighbor Gaussian process models for large geostatistical datasets}.
\newblock {\em Journal of the American Statistical Association}, 111(514):800--812.

\bibitem[Drennan et~al., prep]{drennan2024covtmprep}
Drennan, D., Wiemann, P., and Katzfuss, M. (in prep).
\newblock {Generative modeling of conditional spatial distributions via autoregressive Gaussian processes}.
\newblock {\em In preparation}.

\bibitem[Goodfellow et~al., 2016]{Goodfellow2016}
Goodfellow, I., Bengio, Y., and Courville, A. (2016).
\newblock {\em {Deep Learning}}.
\newblock MIT Press.

\bibitem[Gr{\"{a}}ler, 2014]{Graler2014}
Gr{\"{a}}ler, B. (2014).
\newblock {Modelling skewed spatial random fields through the spatial vine copula}.
\newblock {\em Spatial Statistics}, 10:87--102.

\bibitem[Haugen et~al., 2019]{Haugen2019}
Haugen, M.~A., Stein, M.~L., Sriver, R.~L., and Moyer, E.~J. (2019).
\newblock {Future climate emulations using quantile regressions on large ensembles}.
\newblock {\em Advances in Statistical Climatology, Meteorology, and Oceanography}, 5:37--55.

\bibitem[Hestness et~al., 2017]{Hestness2017}
Hestness, J., Narang, S., Ardalani, N., Diamos, G., Jun, H., Kianinejad, H., Patwary, M. M.~A., Yang, Y., and Zhou, Y. (2017).
\newblock {Deep learning scaling is predictable, empirically}.
\newblock {\em arXiv:1712.00409}.

\bibitem[Houtekamer and Zhang, 2016]{Houtekamer2016}
Houtekamer, P.~L. and Zhang, F. (2016).
\newblock {Review of the ensemble Kalman filter for atmospheric data assimilation}.
\newblock {\em Monthly Weather Review}, 144(12):4489--4532.

\bibitem[Jurek and Katzfuss, 2022]{Jurek2020}
Jurek, M. and Katzfuss, M. (2022).
\newblock {Hierarchical sparse Cholesky decomposition with applications to high-dimensional spatio-temporal filtering}.
\newblock {\em Statistics and Computing}, 32:15.

\bibitem[Kang and Katzfuss, 2023]{Kang2021}
Kang, M. and Katzfuss, M. (2023).
\newblock {Correlation-based sparse inverse Cholesky factorization for fast Gaussian-process inference}.
\newblock {\em Statistics and Computing}, 33(56):1--17.

\bibitem[Kang et~al., 2024]{Kang2024}
Kang, M., Sch{\"{a}}fer, F., Guinness, J., and Katzfuss, M. (2024).
\newblock {Asymptotic properties of Vecchia approximation for Gaussian processes}.
\newblock {\em arXiv:2401.15813}.

\bibitem[Kashinath et~al., 2021]{Kashinath2021}
Kashinath, K., Mustafa, M., Albert, A., Wu, J.~L., Jiang, C., Esmaeilzadeh, S., Azizzadenesheli, K., Wang, R., Chattopadhyay, A., Singh, A., Manepalli, A., Chirila, D., Yu, R., Walters, R., White, B., Xiao, H., Tchelepi, H.~A., Marcus, P., Anandkumar, A., Hassanzadeh, P., and {Prabhat} (2021).
\newblock {Physics-informed machine learning: Case studies for weather and climate modelling}.
\newblock {\em Philosophical Transactions of the Royal Society A}, 379(20200093).

\bibitem[Katzfuss and Guinness, 2021]{Katzfuss2017a}
Katzfuss, M. and Guinness, J. (2021).
\newblock {A General Framework for Vecchia Approximations of Gaussian Processes}.
\newblock {\em Statistical Science}, 36(1):124--141.

\bibitem[Katzfuss et~al., 2020]{Katzfuss2018}
Katzfuss, M., Guinness, J., Gong, W., and Zilber, D. (2020).
\newblock {Vecchia approximations of Gaussian-process predictions}.
\newblock {\em Journal of Agricultural, Biological, and Environmental Statistics}, 25(3):383--414.

\bibitem[Katzfuss and Sch{\"{a}}fer, 2024]{Katzfuss2021}
Katzfuss, M. and Sch{\"{a}}fer, F. (2024).
\newblock {Scalable Bayesian transport maps for high-dimensional non-Gaussian spatial fields}.
\newblock {\em Journal of the American Statistical Association}, 119(546):1409--1423.

\bibitem[Katzfuss et~al., 2016]{Katzfuss2015b}
Katzfuss, M., Stroud, J.~R., and Wikle, C.~K. (2016).
\newblock {Understanding the ensemble Kalman filter}.
\newblock {\em The American Statistician}, 70(4):350--357.

\bibitem[Kay et~al., 2015]{Kay2015}
Kay, J.~E., Deser, C., Phillips, A., Mai, A., Hannay, C., Strand, G., Arblaster, J.~M., Bates, S.~C., Danabasoglu, G., Edwards, J., Holland, M., Kushner, P., Lamarque, J.~F., Lawrence, D., Lindsay, K., Middleton, A., Munoz, E., Neale, R., Oleson, K., Polvani, L., and Vertenstein, M. (2015).
\newblock {The Community Earth System Model (CESM) Large Ensemble Project: A community resource for studying climate change in the presence of internal climate variability}.
\newblock {\em Bulletin of the American Meteorological Society}, 96(8):1333--1349.

\bibitem[Kidd and Katzfuss, 2022]{Kidd2020}
Kidd, B. and Katzfuss, M. (2022).
\newblock {Bayesian nonstationary and nonparametric covariance estimation for large spatial data (with discussion)}.
\newblock {\em Bayesian Analysis}, 17(1):291--351.

\bibitem[Kobyzev et~al., 2021]{Kobyzev2020normalizing}
Kobyzev, I., Prince, S.~J., and Brubaker, M.~A. (2021).
\newblock Normalizing flows: An introduction and review of current methods.
\newblock {\em IEEE Transactions on Pattern Analysis and Machine Intelligence}, 43(11):3964--3979.

\bibitem[Marzouk et~al., 2016]{Marzouk2016}
Marzouk, Y.~M., Moselhy, T., Parno, M., and Spantini, A. (2016).
\newblock {Sampling via measure transport: An introduction}.
\newblock In Ghanem, R., Higdon, D., and Owhadi, H., editors, {\em Handbook of Uncertainty Quantification}. Springer.

\bibitem[Mescheder et~al., 2018]{Mescheder2018}
Mescheder, L., Geiger, A., and Nowozin, S. (2018).
\newblock {Which training methods for GANs do actually converge?}
\newblock In {\em International Conference on Machine Learning}, pages 3481--3490.

\bibitem[Nychka et~al., 2018]{Nychka2018}
Nychka, D.~W., Hammerling, D.~M., Krock, M., and Wiens, A. (2018).
\newblock {Modeling and emulation of nonstationary Gaussian fields}.
\newblock {\em Spatial Statistics}, 28:21--38.

\bibitem[Rosenblatt, 1952]{Rosenblatt1952}
Rosenblatt, M. (1952).
\newblock {Remarks on a multivariate transformation}.
\newblock {\em The Annals of Mathematical Statistics}, 23(3):470--472.

\bibitem[Sch{\"{a}}fer et~al., 2021a]{Schafer2020}
Sch{\"{a}}fer, F., Katzfuss, M., and Owhadi, H. (2021a).
\newblock {Sparse Cholesky factorization by Kullback-Leibler minimization}.
\newblock {\em SIAM Journal on Scientific Computing}, 43(3):A2019--A2046.

\bibitem[Sch{\"{a}}fer et~al., 2021b]{Schafer2017}
Sch{\"{a}}fer, F., Sullivan, T.~J., and Owhadi, H. (2021b).
\newblock {Compression, inversion, and approximate PCA of dense kernel matrices at near-linear computational complexity}.
\newblock {\em Multiscale Modeling {\&} Simulation}, 19(2):688--730.

\bibitem[Sch\"{a}fer et~al., 2021]{Schafer2021b}
Sch\"{a}fer, F., Sullivan, T.~J., and Owhadi, H. (2021).
\newblock Compression, inversion, and approximate pca of dense kernel matrices at near-linear computational complexity.
\newblock {\em Multiscale Modeling \& Simulation}, 19(2):688--730.

\bibitem[Stein, 2011]{Stein2011}
Stein, M.~L. (2011).
\newblock {When does the screening effect hold?}
\newblock {\em Annals of Statistics}, 39(6):2795--2819.

\bibitem[Stein et~al., 2004]{Stein2004}
Stein, M.~L., Chi, Z., and Welty, L. (2004).
\newblock {Approximating likelihoods for large spatial data sets}.
\newblock {\em Journal of the Royal Statistical Society: Series B}, 66(2):275--296.

\bibitem[Vecchia, 1988]{Vecchia1988}
Vecchia, A. (1988).
\newblock {Estimation and model identification for continuous spatial processes}.
\newblock {\em Journal of the Royal Statistical Society, Series B}, 50(2):297--312.

\bibitem[Wiemann and Katzfuss, 2023]{Wiemann2023BayesianFields}
Wiemann, P. F.~V. and Katzfuss, M. (2023).
\newblock {Bayesian nonparametric generative modeling of large multivariate non-Gaussian spatial fields}.
\newblock {\em Journal of Agricultural, Biological and Environmental Statistics}, 28(4):597--617.

\bibitem[Wiens et~al., 2020]{Wiens2020}
Wiens, A., Nychka, D.~W., and Kleiber, W. (2020).
\newblock {Modeling spatial data using local likelihood estimation and a Mat{\'{e}}rn to spatial autoregressive translation}.
\newblock {\em Environmetrics}, 31(6):1--15.

\bibitem[Zhang and Katzfuss, 2022]{Zhang2022}
Zhang, J. and Katzfuss, M. (2022).
\newblock {Multi-scale Vecchia approximations of Gaussian processes}.
\newblock {\em Journal of Agricultural, Biological and Environmental Statistics}, 27:440--460.

\bibitem[Zilber and Katzfuss, 2021]{Zilber2019}
Zilber, D. and Katzfuss, M. (2021).
\newblock {Vecchia-Laplace approximations of generalized Gaussian processes for big non-Gaussian spatial data}.
\newblock {\em Computational Statistics {\&} Data Analysis}, 153:107081.

\end{thebibliography}

\end{document}